\documentclass{article}

\usepackage{arxiv}

\usepackage[utf8]{inputenc} 
\usepackage[T1]{fontenc}    
\usepackage{hyperref}       
\usepackage{url}            
\usepackage{booktabs}       
\usepackage{amsfonts}       
\usepackage{nicefrac}       
\usepackage{microtype}      
\usepackage{lipsum}		
\usepackage{graphicx}
\usepackage{natbib}
\usepackage{doi}

\usepackage{fluff} 
\usepackage[ruled]{algorithm2e}
\usepackage{enumitem}
\usepackage{cleveref}

\usepackage{fluff} 
\usepackage{enumitem}
\usepackage{cleveref}

\newtheorem{problem}{\bf Open Problem}

\newcommand{\ahzr}{\textit{RHZ}}

\renewcommand{\epsilon}{\varepsilon}

\newcommand{\indeg}{\text{in-deg}}
\newcommand{\outdeg}{\text{out-deg}}
\newcommand{\tfor}{\text{ for }}
\newcommand{\gcircuit}{\textsc{Gcircuit}}
\newcommand{\threshold}{\textsc{Threshold}}

\DeclareMathOperator*{\argmax}{arg\,max}

\title{Hardness of approximate Hylland-Zeckhauser equilibria}

\author{Mark Braverman\thanks{Princeton University} \And Jingyi Liu\thanks{Princeton University} \And Eric Xue\thanks{Princeton University} \And Chenghan Zhou\thanks{Stanford University}}

\date{}

\renewcommand{\headeright}{}
\renewcommand{\undertitle}{}

\hypersetup{
pdftitle={Hardness of approximate Hylland-Zeckhauser equilibria},
pdfsubject={cs.GT, cs.CC},
pdfauthor={Mark Braverman, Jingyi Liu, Eric Xue, Chenghan Zhou},
pdfkeywords={Hylland-Zeckhauser equilibrium; approximate equilibria; CEEI; PPAD-hardness; generalized circuits; PCP-for-PPAD; one-sided allocation},
}

\begin{document}
\maketitle

\begin{abstract}
In this paper, we investigate the computational hardness of finding fractional allocations to unit-demand players using competitive equilibria from equal incomes (CEEI), where we allow a small constant error in players' response to market prices (also known as an approximate Hylland-Zeckhauser equilibrium). We show that assuming the  $\mathbf{(\varepsilon,\delta)}$-$\gcircuit$ problem is PPAD-hard (the PCP-for-PPAD conjecture), finding an approximate HZ equilibrium is also PPAD-hard. This result provides additional motivation for trying to prove the PCP-for-PPAD conjecture as a tool for obtaining robust computational hardness results about markets. Further, we introduce a natural restriction on approximate HZ equilibria, where players' bundles may still only be approximately optimal given the prices, but may not contain positive-price items for which the player has zero utility. We show unconditionally that there exists a constant $\epsilon$ such that finding a {\em restricted} $\epsilon$-HZ equilibrium is PPAD-hard.
\end{abstract}



\section{Introduction}

\paragraph{Allocation problems without money.} In this paper we focus on allocating $n$ divisible goods to $n$ unit-demand players. A canonical application is assigning students to schools: each student receives a probability distribution over schools, which is then sampled using a lottery. Each student's utility is their expected utility over the lottery. 

Formally, there are $n$ players and $n$ items. Each player $i$ has utility $u_{ij}\ge 0$ for item $j$. A valid allocation $\{X_{ij}\}_{i\in[n],j\in[n]}$ satisfies capacity constraints and unit-demand constraints:
$$
\forall i,j~X_{ij}\ge0;~~~~~~~\forall j~\sum_{i} X_{ij}=1; ~~~~~~\forall i~\sum_{j} X_{ij}=1
$$
The utility of player $i$ under such an allocation is 
$$
U_i(X)=\sum_j u_{ij}\cdot X_{ij}
$$

In a setting {\em with} money, this allocation problem leads to a unit-demand auction, which is very well understood \citep{Shapley71, Leonard83}. The auction's VCG prices can be easily computed by a linear program. Moreover, there is a natural tatonnement dynamics that converges to those prices \citep{demange1986multi}. 

\paragraph{Pseudo-market-based allocation.}
In the setting {\em without} money, a natural approach is to use a token currency to produce an allocation. Such an allocation is called a competitive equilibrium from equal incomes (CEEI) or a Hylland-Zeckhauser (HZ) equilibrium \citep{HZ1979}. 

Formally, an allocation $\{X_{ij}\}$ is supported by an HZ equilibrium if there are item prices $P_j\ge 0$ such that when faced with these prices and a budget constraint each player picks their bundle $X_i$ to maximize expected utility:

\begin{equation}
    \label{eq:HZ1}
\forall i~X_i \in \argmax_y\left\{ \sum_j u_{ij} \cdot y_j ~:~\sum_j y_j =1 \text{ and }\sum_{j} P_j \cdot y_j\le 1\right\}
\end{equation}

Pseudo-market based approaches are not the only way to solve allocation problems without money. There is a rich literature on combinatorial allocation mechanisms such as random serial dictatorship, top trading cycles, and probabilistic serial \citep{Abdulkadiroglu98, Shapley74, Bogomolnaia01}. The pseudo-market based approach has several important benefits. Firstly, it incorporates cardinal utility information, which is invariably lost when considering ranking-based mechanisms. Secondly, it is not hard to see that any CEEI outcome is guaranteed to be Pareto optimal (PO)\footnote{Provided that each player chooses the cheapest bundle that gives them the same maximum utility} and (ex-ante) envy-free (EF) --- these are simple consequences of market efficiency.  

While an allocation being PO+EF does not generally imply that it is supported by a CEEI, additional ``coalition envy-free" conditions imply CEEI \citep{Varian74}. This means that pseudo-market-based equilibria are a natural notion of fairness and efficiency for allocation problems. Moreover, there is a generic reduction that shows that, computationally, the problem of finding a CEEI equilibrium is not more difficult than the problem of ``only" finding an allocation that is Pareto-optimal and envy free \citep{TV24}. 

Finally, as compared to combinatorial approaches, it is much easier to extend market-based approaches to allocation problems to obtain market-based approaches to much more general coordination problems. 

\paragraph{Computing HZ equilibria.} Perhaps the main reason for pseudo-market-based mechanisms not being more widely adopted is that computing equilibrium prices is computationally difficult. While allocation with money can be solved efficiently using simple demand queries, no such algorithm is known for the pseudo-market setting. 

Existence of HZ equilibria is proved using a fixed-point argument, and the exact problem is naturally formulated as a real-valued fixed-point problem: \cite{vazirani2021computational} showed that exact HZ lies in FIXP and exhibited instances with irrational equilibria. 
For the finite-precision problem, they showed that approximate HZ lies in PPAD, and \cite{CCPY21} proved the matching hardness result for high-precision approximation, showing that computing an $\varepsilon$-approximate HZ equilibrium is PPAD-hard even when $\varepsilon$ is inverse-polynomially small. Hence high-precision approximate HZ is PPAD-complete.

These results rule out efficient exact or near-exact computation under standard complexity assumptions, but they leave open the possibility of coarser approximations. This distinction is important:  when $\varepsilon$ is sufficiently small, errors are negligible, even when accumulated over the entire market. Thus the reduction from a PPAD-hard fixed-point problem to finding an HZ equilibrium is not affected by such errors. In a constant-error approximate HZ equilibrium, agents may make non-negligible mistakes in their responses to prices, and these mistakes can interact with market clearing in ways that are not present in the high-precision regime. This motivates the main question of the paper:

\smallskip
\noindent
{\bf Main Problem.} Does there exist a constant error $\varepsilon>0$ such that computing an $\varepsilon$-HZ equilibrium is PPAD-hard?

\smallskip

Note that -- for now -- we do not give a precise definition of an $\varepsilon$-relaxation notion, as definition details may affect the answer.\footnote{Very recent independent work of ~\cite{yanliu2026approxhz} reports a polynomial-time 1/e-approximation algorithm for HZ under the definition of \cite{vazirani2021computational}. Our first result is complementary: under the same definition, we study the hardness of obtaining sufficiently accurate constant-error approximations and show that PCP-for-PPAD implies PPAD-hardness in this regime. Our second result gives an unconditional hardness theorem for approximate HZ under another natural, restricted definition.}

\subsection{The approximate fixed-point conjecture(s)}

The goal of this section is to explain the different flavors of $\varepsilon$-fixed point conjectures/theorems. There are several versions of approximate fixed point problems. Some of these are known to be PPAD-hard, although not the ones needed to prove hardness results for approximate market (or pseudo-market) equilibria. 

A particularly useful PPAD-complete problem is that of {\em generalized circuits} (or $\gcircuit$). One constructs a circuit with gates and wires connecting to input and output nodes. The circuit is allowed to have loops (thus, gates can be thought of as constraints on the nodes). The goal of the $\gcircuit$ problem is to find a fixed point in $[0,1]^V$, where $V$ is the number of nodes. Note that as long as the gates' functions are continuous, Brouwer's Fixed Point Theorem guarantees that a fixed point exists. To avoid complications with real numbers, one should consider approximate gates. It turns out that this is sufficient for application to computational complexity in game theory. Informally:

\smallskip
\noindent 
{\bf The $\mathbf{\varepsilon}$-$\gcircuit$ problem}: assign a value from $[0,1]$ to every node such that the value of each gate output is within an additive $\varepsilon$ from the value prescribed by the gate computing it. 
\smallskip

$\varepsilon$-$\gcircuit$ was introduced by \cite{CDT09-gcircuit} as an intermediate problem to show $(1/poly(n))$-approximate 2-Nash is PPAD-hard. They showed that $(1/poly(n))$-$\gcircuit$ is PPAD-complete by a reduction from a discrete Brouwer problem\footnote{The $\varepsilon$-$\gcircuit$ problem can be reduced to finding an approximate fixed point of a given continuous map from $[0,1]^n\rightarrow[0,1]^n$, and thus is in PPAD.}. Subsequent work by \cite{Rub18-inapprox} improved this result and showed that $\epsilon$-$\gcircuit$ is PPAD-complete for some constant $\epsilon>0$ even for generalized circuits of fan-out 2. 
\cite{Rub18-inapprox}'s result on $\epsilon$-$\gcircuit$ is used to show constant inapproximability of other problems \citep{SSB17,FRFGZ18,PP21,GHI+22,FRGH+23}. In an effort to nail down the constant $\epsilon$ in the inapproximability results, \cite{DFHM22-Pure-Circuit} introduced the Pure-Circuit problem, which can be thought of as taking $\epsilon\to 1$ for $\epsilon$-$\gcircuit$, and used it to further tighten results on the hardness of computing Nash equilibria. By a reduction from the Pure-Circuit problem,  \cite{DFHM22-Pure-Circuit} showed that $\epsilon$-$\gcircuit$ is PPAD-complete for all $\epsilon < 0.1$.

 The $(\varepsilon,\delta)$-$\gcircuit$ problem was initially proposed by \cite{BPR15-PCP}.
There are some nuances in how exactly the problem should be defined, but the version we will use is a direct relaxation of the $\varepsilon$-$\gcircuit$ problem, where gates have fan-out $\le 2$, and a $\delta$-fraction of gates is allowed to behave arbitrarily:

\smallskip
\noindent 
{\bf The $\mathbf{(\varepsilon,\delta)}$-$\gcircuit$ problem}: assign a value from $[0,1]$ to every node such that {\em for all but a $\delta$-fraction of gates}, the value of the gates' output is within an additive $\varepsilon$ from the value prescribed based on its inputs. 
\smallskip

Since any assignment satisfying all gates of an $\varepsilon$-$\gcircuit$ instance also satisfies all but a $\delta$-fraction of them, $(\varepsilon,\delta)$-$\gcircuit$ reduces to $\varepsilon$-$\gcircuit$ and hence lies in \(\mathrm{PPAD}\). \cite{BPR15-PCP} conjectured that the $(\varepsilon,\delta)$-$\gcircuit$ problem remains PPAD-hard. This is referred to as the PCP-for-PPAD conjecture, since $(\epsilon,\delta)$-$\gcircuit$ can be interpreted as a PCP formulation of the $\epsilon$-$\gcircuit$ problem. They showed that this conjecture implies finding prices and income assignment for the course allocation problem with low market clearing error and near-optimal Gini index is PPAD-complete. Interestingly, the original $\epsilon$-$\gcircuit$ problem was used to show unconditionally that there exists $\beta>0$ such that the closely related solution concept $(\alpha^*,\beta)$-A-CEEI (as defined in \cite{Budish11}) is PPAD-complete where $\alpha^*$ is a known function of input.\footnote{The equilibrium concept $(\alpha,\beta)$-approximate CEEI (or A-CEEI) was first proposed by \cite{Budish11} to solve the allocation problem of indivisible goods without money where each agent is given a fake budget of $B_i$ (which can be unequal) and items are priced so that when each agent chooses their favorite integral bundle of items, market clears with a small clearing error $\alpha$, and the inequality in agents' budgets is bounded by $\beta$. However, this is not to be confused with the approximate CEEI (or approximate HZ equilibrium) defined in our setting, where items are divisible and each agent has unit-demand (alternatively, the allocation can be thought of as bundles of probability shares of different indivisible items). We note that although our main result also connects the conjecture to the computational hardness of an approximate-CEEI (i.e. an approximate HZ), our fractional setting makes the reduction much more difficult as each item can now be allocated to many agents, each getting a tiny fraction and it is harder to reason about prices and allocations in any equilibrium.} 

Our main goal is to tighten the link between (hardness of) approximate market equilibria and well-established conjectures around fixed point computations. 

As a secondary goal, we contribute to establishing the $(\varepsilon,\delta)$-$\gcircuit$ problem as the ``right" problem to reduce from when establishing the hardness of market equilibria. This should encourage more work on settling the complexity of this problem. We expect either resolution to be beneficial for our understanding of the hardness of approximate market equilibria: if the $(\varepsilon,\delta)$-$\gcircuit$ problem turns out to be computationally difficult (thus establishing PCP-for-PPAD), then it could provide a general source of robust reductions---the hardness of exact market equilibria will likely translate into hardness of approximate market equilibria---as it does in the present paper.  On the other hand, a surprising, efficient algorithm for the 
$(\epsilon,\delta)$-$\gcircuit$ problem
would suggest that new algorithms for approximate HZ and related market-equilibrium problems may be possible. 

Recent independent work on Fisher markets gives further evidence that $(\epsilon,\delta)$-$\gcircuit$  is the right fixed-point problem for robust approximate market equilibria. \cite{deligkas2024constantfisher} prove constant inapproximability for Fisher markets with separable piecewise-linear concave utilities when the approximation is in market clearing and buyers still receive exactly optimal bundles. A subsequent work of ~\cite{deligkas2026approxoptimalfisher} studies the more closely related relaxation in which buyers receive approximately optimal bundles, and shows that, assuming the PCP-for-PPAD conjecture, there exist constants $\varepsilon$ and $\delta$ such that computing an approximate market equilibrium where buyers receive $(1-\delta)$-optimal bundles and every good clears with $\varepsilon$ error is PPAD-hard, even for Fisher markets with equal budgets and linear capped utilities.\footnote{We emphasize that our results on the hardness of computing approximate HZ equilibria are not corollaries of the Fisher-market hardness results. Although HZ can be viewed as a CEEI-style market and has linear utilities, it is not an ordinary SPLC Fisher market. In a standard Fisher market, each buyer chooses a utility-maximizing bundle subject only to a budget constraint and nonnegativity. In HZ, each agent chooses a lottery subject simultaneously to a budget constraint and the unit-demand constraint ($\sum_j x_{ij}=1$). Linear capped utilities in Fisher markets can model satiation for individual goods, but they do not impose this cross-good unit-demand constraint. Our reduction therefore works directly with the HZ allocation polytope and the normalized-price structure of HZ, rather than reducing through Fisher markets.} They also prove a converse for a broad class of reducible Fisher markets: PPAD-hardness for constant approximate optimality in that class would imply PCP-for-PPAD.

\subsection{Our results and discussion}

All our results deal with the computational hardness of computing approximate HZ equilibria for allocating goods without money. As in other ``approximate" settings (such as approximate Nash Equilibria), the approximation refers to the extent to which each player best responds to market prices. An allocation $X$ is supported by an approximate $\varepsilon$-HZ equilibrium with prices $P$ if budget and optimality conditions \eqref{eq:HZ1} hold up to $\varepsilon$:
\begin{align}
    &\forall i~\sum_j P_j X_{ij} \le 1+\varepsilon; \nonumber \\
    &\forall i~\sum_{j} u_{ij} \cdot X_{ij}  \geq \max_{y\geq 0}\left\{ \sum_j u_{ij} \cdot y_j ~:~\sum_j y_j =1 \text{ and }\sum_{j} P_j \cdot y_j\le 1\right\}-\varepsilon.
\label{eq:HZ2}
\end{align}

Our first result is on the hardness of computing 
an $\epsilon$-HZ equilibrium. We show that there is a polynomial time reduction from the $(\epsilon,\delta)$-$\gcircuit$ problem to the problem of computing an $\epsilon$-HZ:

\smallskip
\noindent
{\bf Theorem \ref{thm:conj-implies-eps-hz-ppad} (restated).} For any {positive} constants $(\varepsilon,\delta)$, there is a polynomial-time reduction from the $(\epsilon,\delta)$-$\gcircuit$ problem to the problem of computing an $\epsilon'$-HZ for some other constant $\epsilon'$.

Assuming the PCP-for-PPAD conjecture, this implies that computing an approximately CEEI solution for unit-demand allocation is PPAD-hard. Even without this conjecture, it ties the problem of finding approximate equilibria to a well-stated open problem in approximate fixed-point computation.

\smallskip

Our second result aims to prove that $\epsilon$-HZ equilibria are difficult to compute without relying on new conjectures. To this end, we introduce a new notion of an {\bf approximate HZ equilibrium with restriction} (Definition~\ref{def:epsRHZ}). Informally, an allocation is an $\varepsilon$-RHZ if (1) each agent chooses a feasible bundle that maximizes their utility up to an additive $\varepsilon$ {using a relaxed budget of up to $1+\varepsilon$}; and (2) no agent ever chooses any positive-price items for which they have $0$ utility. In other words, in addition to condition \eqref{eq:HZ2} we require that
\begin{equation}
    \label{eq:HZ3}
    \forall i,j~~u_{ij}=0 \land P_j>0 \Rightarrow X_{ij}=0.
\end{equation}
Consider a market where each player only considers a subset of items (for example, a student may only be interested in a school in their neighborhood, or in a course in their major). 
Intuitively, restriction \eqref{eq:HZ3} says that while a player may not respond optimally to a given vector of prices, they will never even consider items outside of their demand set, and thus we need not consider solutions supporting such responses. 

Note that without such a restriction, any solution that removes an $O(\epsilon)$-fraction of the goods from the market can be implemented as an $\epsilon$-HZ solution (by having each player use $\epsilon$-fraction of its budget to buy and burn these items). The restriction aims to curtail such arbitrary ``long-range" effects within the market. 

It turns out that with condition \eqref{eq:HZ3} the problem of finding $\varepsilon$-RHZ becomes PPAD-hard:

\smallskip
\noindent
{\bf Theorem \ref{thm:ppad-constant}.}
    Finding an $\epsilon$-$\ahzr$ with restriction is PPAD-hard for some constant $\epsilon>0$.

\smallskip

Technically, the proof of Theorem~\ref{thm:ppad-constant} is very similar to the hardness proof for polynomially-small $\varepsilon$
in \cite{CCPY21}. (We defer the proof of this theorem to Appendix \ref{sec:hard with restriction}.) Our main contribution in this result is a conceptual one: introducing a natural restriction for a market-based equilibrium that makes constant-approximation PPAD-hard. We believe that such a restriction will be useful in other market-based scenarios where the goal is to allow players to approximately best respond to the market without creating new arbitrary demands. 

\subsection{Open questions}
Our results concern price-supported relaxations of HZ equilibria. A natural next question is whether similar hardness phenomena hold for allocation-only fairness and efficiency guarantees, where the allocation need not be supported by prices. \cite{TV24} show that computing an exact Pareto-optimal and envy-free (PO+EF) allocation is PPAD-hard, via a reduction from inverse-polynomial approximate HZ. They also obtain positive approximation results for coarser notions of approximate envy-freeness.\footnote{\cite{TV24} give a polynomial-time algorithm to compute a $(2+\varepsilon)$-approximately envy-free and (exactly) Pareto-optimal allocation using a multiplicative notion of approximate envy-freeness.} Thus the remaining question is not whether every relaxation of PO+EF is hard, but whether sufficiently fine relaxations remain computationally intractable.

One natural additive relaxation is the following.

\begin{problem}
Is there a constant $\epsilon>0$ for which it is PPAD-hard to find an $\epsilon$-relaxed Pareto Optimal + Envy Free allocation? That is, every player experiences at most $\epsilon$ envy, and there is no Pareto improvement in which everyone is at least as well off, and aggregate welfare increases by at least $\epsilon\cdot n$?
\end{problem}

Note that the problem above is a priori stronger than the hardness of $\epsilon$-HZ for a constant $\epsilon>0$. Therefore, a more reasonable conjecture would be to  look for a reduction from the PCP-for-PPAD problem:

\begin{problem}
Are there constants $\epsilon,\delta,\epsilon'>0$ for which there is a polynomial time reduction from the $\mathbf{(\varepsilon,\delta)}$-$\gcircuit$ problem to the problem of finding an $\epsilon'$-relaxed Pareto Optimal + Envy Free allocation?
\end{problem}

\section{Preliminaries}

A one-sided allocation problem consists of $n$ unit-demand agents and $n$ divisible goods, each of which has unit supply.
We use $u_{ij}\in [0,1]$ to denote agent $i$'s utility for one unit of good $j$.
The HZ scheme is a market-based approach for allocating goods to agents.
More specifically, in an HZ market, each agent is given one unit of ``fake'' currency with no value outside of the market that they can use to buy the goods.
An $\varepsilon$-HZ equilibrium consists of an allocation $x \in \RR^{n \times n}_{\geq 0}$ and item prices $p \in\RR^n_{\geq 0}$ that satisfy the following conditions.
For all agents $i$, we use $x_i$ to denote her allocation.

\begin{definition}[$\epsilon$-HZ, \citep{HZ1979, vazirani2021computational, CCPY21}]
Given $\epsilon>0$, a pair $(x,p)$, where $x\in \RR^{n\times n}_{\geq 0}$ and $p\in \RR^{n}_{\geq 0}$ is an $\epsilon$-approximate HZ equilibrium of an HZ market $M$ if:
\begin{enumerate}
    \item Unit supply: $\sum_{i\in [n]}x_{ij}=1, \forall j$
    \item Unit demand: $\sum_{j\in [n]}x_{ij}=1, \forall i$
    \item Normalized prices: $\min_{j\in [n]}p_j=0$
    \item Budget is approximately 1: $\sum_{j\in [n]}p_j x_{ij}\leq 1+\epsilon, \forall i$
    \item Bundle is approximately optimal: $\sum_{j\in [n]}u_{ij}x_{ij}\geq val_p(i)-\epsilon, \forall i$
\end{enumerate}
where $val_p(i) = \max_{y\geq 0}\left\{ \sum_{j\in [n]} u_{ij} y_j ~:~\sum_{j\in [n]} y_j =1 \text{ and }\sum_{j\in [n]} p_j  y_j\le 1\right\}$ is the value of the best affordable bundle for agent $i$ under unit demand.
\end{definition}
\begin{remark*}
   \cite{CCPY21} observed that the normalization of prices is necessary since HZ equilibria are invariant to the following transformation of the prices: let $(x,p)$ be an HZ equilibrium. Then for any $0<r\leq \min\{1/(1-p_j)\mid p_j<1\}$, we can rescale the prices $p$ to $p'$ where $p'_j-1=r(p_j-1)$ and $(x,p')$ is also an HZ equilibrium. If we do not normalize the prices, then condition 4 can be trivially satisfied by transforming the prices so that all prices are close to 1. With this observation, they require the normalization in condition (3) that the minimum item price is 0. We also follow this convention. 
\end{remark*}

Next we introduce the generalized circuits problem and its relaxation to $(\epsilon,\delta)$-\gcircuit. 

\begin{definition}
[Generalized circuits, \citep{CDT09-gcircuit}] 
A \textit{generalized circuit} $S$ is a pair $(V, \mathcal{T})$, where $V$ is a set of nodes and $\mathcal{T}$ is a collection of gates. Every gate $T \in \mathcal{T}$ is a 5-tuple $T = G(\zeta\mid v_1,v_2\mid v)$, in which $G\in\{G_\zeta, G_{\times\zeta}, G_=, G_+, G_-, G_<, G_{\lor}, G_{\land}, G_{\lnot}\}$ is the type of the gate; $\zeta \in \mathbb{R} \cup \{nil\}$ is an (optional) real parameter; $v_1, v_2 \in V \cup \{nil\}$ are the first and second input nodes of the gate (one or both of them may be missing); and $v \in V$ is the output node; no two distinct gates have the same output node.
\end{definition}

\begin{definition}[$(\epsilon,\delta)$-\gcircuit, \citep{BPR15-PCP}] Given a generalized circuit $S = (V,\mathcal{T})$, we say that an assignment $x: V \to [0,1]$ $(\epsilon, \delta)$-\textit{approximately satisfies} $S$ if for all but a $\delta$-fraction of the gates, $x$ satisfies the corresponding constraints in Table \ref{tbl:gcircuit}.

\begin{table}[ht]
\centering
\begin{tabular}{|c|c|}
\hline
\textbf{Gate} & \textbf{Constraint} \\ \hline
$G_\zeta(\alpha \parallel a)$ & $\mathbf{x}[a] = \alpha \pm \epsilon$ \\ \hline
$G_{\times\zeta}(\alpha \mid a \mid b)$ & $\mathbf{x}[b] = \alpha \cdot \mathbf{x}[a] \pm \epsilon$ \\ \hline
$G_=(\mid a \mid b)$ & $\mathbf{x}[b] = \mathbf{x}[a] \pm \epsilon$ \\ \hline
$G_+( \mid a,b \mid c)$ & $\mathbf{x}[c] = \min(\mathbf{x}[a] + \mathbf{x}[b], 1) \pm \epsilon$ \\ \hline
$G_-( \mid a,b \mid c)$ & $\mathbf{x}[c] = \max(\mathbf{x}[a] - \mathbf{x}[b], 0) \pm \epsilon$ \\ \hline
$G_<(\mid a,b \mid c)$ & 
\begin{tabular}[c]{@{}c@{}}
$\mathbf{x}[c] = \begin{cases} 
1 \pm \epsilon & \text{if } \mathbf{x}[a] < \mathbf{x}[b] - \epsilon \\
0 \pm \epsilon & \text{if } \mathbf{x}[a] > \mathbf{x}[b] + \epsilon 
\end{cases}$
\end{tabular} \\ \hline
$G_{\lor}(\mid a,b \mid c)$ & 
\begin{tabular}[c]{@{}c@{}}
$\mathbf{x}[c] = \begin{cases} 
1 \pm \epsilon & \text{if } \mathbf{x}[a] = 1 \pm \epsilon \text{ or } \mathbf{x}[b] = 1 \pm \epsilon \\
0 \pm \epsilon & \text{if } \mathbf{x}[a] = 0 \pm \epsilon \text{ and } \mathbf{x}[b] = 0 \pm \epsilon 
\end{cases}$
\end{tabular} \\ \hline
$G_{\land}(\mid a,b \mid c)$ & 
\begin{tabular}[c]{@{}c@{}}
$\mathbf{x}[c] = \begin{cases} 
1 \pm \epsilon & \text{if } \mathbf{x}[a] = 1 \pm \epsilon \text{ and } \mathbf{x}[b] = 1 \pm \epsilon \\
0 \pm \epsilon & \text{if } \mathbf{x}[a] = 0 \pm \epsilon \text{ or } \mathbf{x}[b] = 0 \pm \epsilon 
\end{cases}$
\end{tabular} \\ \hline
$G_{\lnot}(\mid a \mid b)$ & 
\begin{tabular}[c]{@{}c@{}}
$\mathbf{x}[b] = \begin{cases} 
1 \pm \epsilon & \text{if } \mathbf{x}[a] = 0 \pm \epsilon \\
0 \pm \epsilon & \text{if } \mathbf{x}[a] = 1 \pm \epsilon 
\end{cases}$
\end{tabular} \\ \hline
\end{tabular}
\caption{Generalized Circuit Constraints \citep{CDT09-gcircuit, Rub18-inapprox}}
\label{tbl:gcircuit}
\end{table}
    
\end{definition}

\section{Conditional Hardness of Approximating HZ Equilibrium without Restriction}
\label{sec: conditional hardness without restriction}

\cite{BPR15-PCP} conjectured the generalized circuit problem is PPAD-hard to approximate even allowing some constant fraction of the gates to be corrupted. They showed that this conjecture implies computational hardness for computing Nash Equilibria and CEEIs. We show in this paper that this conjecture also implies computing an $\epsilon$-approximate HZ equilibrium is PPAD-hard for some constant $\epsilon>0$. 

\begin{conjecture}[\cite{BPR15-PCP}, Conjecture 2]
\label{conj:gcircuit}
There exist constants $\epsilon,\delta>0$ such that $(\epsilon,\delta)$-$\gcircuit$ is PPAD-hard.
\end{conjecture}

Now we are ready to state our main theorem of this section.
\begin{theorem}
\label{thm:conj-implies-eps-hz-ppad}
Assuming Conjecture $\ref{conj:gcircuit}$, finding an $\epsilon$-HZ equilibrium is PPAD-hard for some constant $\epsilon>0$. 
\end{theorem}

We show this by a reduction from the $(\epsilon,\delta)$-$\gcircuit$ problem to an intermediate problem called $(\kappa, \delta')$-$\threshold$ problem and then a reduction from the threshold problem to the $\epsilon'$-HZ problem. The $(\kappa,\delta)$-$\threshold$ problem is a relaxation of the Threshold game in  \cite{PP21}. It only requires approximate satisfaction of a large fraction of the gates. 

\begin{definition}[$(\kappa,\delta)$-\threshold]
Given a threshold game defined on the directed graph $H=(V,E)$ and threshold $t$, a strategy profile $x=(x_u)_{u\in V}\in [0,1]^{|V|}$ is an $(\kappa, \delta)$-equilibrium of the threshold game if for at least $(1-\delta)$-fraction of the vertices, we have
    \begin{align*}
        x_v = 
        \begin{cases}
            [0,\kappa] &\sum_{u\in N_v}x_u>t+\kappa\\
            [1-\kappa,1] &\sum_{u\in N_v}x_u<t-\kappa\\
            [0,1] &\sum_{u\in N_v}x_u\in [t-\kappa, t+\kappa].\\
        \end{cases}
    \end{align*}
where $N_v$ is the set of vertices $u$ with incoming edges $(u,v)\in E$ to vertex $v$.
\end{definition}

\begin{lemma}
\label{lem:eps-delta-threshold-ppad}
    Assuming Conjecture \ref{conj:gcircuit}, $(\kappa,\delta)$-$\threshold$ is PPAD-hard for some $\kappa, \delta>0$ and directed graph $H$ with bounded in-degree and out-degree.
\end{lemma}

\begin{proof}
Conjecture \ref{conj:gcircuit} says that $(\epsilon,\delta)$-$\gcircuit$ is PPAD-hard, so it suffices to show a reduction from $(\epsilon,\delta)$-$\gcircuit$ to $(\kappa,\delta')$-$\threshold$.
\cite{PP21} give a reduction from $\epsilon$-$\gcircuit$ to $\kappa$-$\threshold$. Their proof uses an alternative but equivalent formulation of the generalized circuit problem. The alternative gates are shown in Table \ref{tbl:gcircuit2}. Consider an instance of the $(\epsilon,\delta)$-$\gcircuit$ problem $(V,\mathcal{T})$. Without loss of generality, we can assume the generalized circuit has fan-out bounded by some constant (see Appendix \ref{append:const fan-out} for a detailed reduction). Let $n=|\mathcal{T}|$ be the number of gates. Then we can construct an instance of the threshold game $H(V,\mathcal{T}),t=1/2$ using the gadget construction in \cite{PP21}, which uses a subgraph consisting of a constant number of vertices to simulate each gate in the generalized circuit. The constructed directed graph $H$ has in-degree and out-degree bounded by some constant and the number of vertices in the graph is $c_2 n$ for some constant $c_2$. Let $(x_v)_{v\in V(H)}$ be an $(\kappa,\delta')$-equilibrium of the threshold game. Then only $\delta'$-fraction of the vertices are not satisfied. We delete any gate in the generalized circuit instance whose gadget in $H$ has at least one unsatisfied vertex. Since each unsatisfied vertex of $V(H)$ appears in at most one gate gadget (it is either a vertex added to simulate some gate or a vertex in $V$ that is the output node of the gate), we will delete at most $c_2n\delta'$ gates. Thus there are $(1-c_2\delta')n$ gates whose constraints are $\epsilon$-approximately satisfied by $(x_v)_{v\in V(H)}$. Set $c_2\delta'<\delta$, then the result follows.
\end{proof}

\begin{table}[h]
\centering
\begin{tabular}{|c|c|}
\hline
\textbf{Gate} & \textbf{Constraint} \\
\hline
$G_{\frac{1}{2}}(v)$ & $x[v] = \frac{1}{2} \pm \epsilon$ \\
\hline
$G_{\times \frac{1}{2}}(\vert v_1 \vert v)$ & $x[v] = \frac{1}{2} \cdot x[v_1] \pm \epsilon$ \\
\hline
$G_=(\vert v_1 \vert v)$ & $x[v] = x[v_1] \pm \epsilon$ \\
\hline
$G_+(\vert v_1, v_2 \vert v)$ & $x[v] = \min\{x[v_1] + x[v_2], \frac{1}{2}\} \pm \epsilon$ \\
\hline
$G_-(\vert v_1, v_2 \vert v)$ & $x[v] = \max\{x[v_1] - x[v_2], 0\} \pm \epsilon$ \\
\hline
$G_{<}(\vert v_1, v_2 \vert v)$ & 
\begin{tabular}{c}
$x[v] = \left\{
\begin{array}{ll}
\frac{1}{2} \pm \epsilon & \text{if } x[v_1] < x[v_2] - \epsilon \\
0 \pm \epsilon & \text{if } x[v_1] > x[v_2] + \epsilon
\end{array}\right.$
\end{tabular} \\
\hline
$G_{\land}(\vert v_1, v_2 \vert v)$ & 
\begin{tabular}{c}
$x[v] = \left\{
\begin{array}{ll}
\frac{1}{2} \pm \epsilon & \text{if } x[v_1] = \frac{1}{2} \pm \epsilon \land x[v_2] = \frac{1}{2} \pm \epsilon \\
0 \pm \epsilon & \text{if } x[v_1] = 0 \pm \epsilon \lor x[v_2] = 0 \pm \epsilon
\end{array}\right.$
\end{tabular} \\
\hline
$G_{\lor}(\vert v_1, v_2 \vert v)$ & 
\begin{tabular}{c}
$x[v] = \left\{
\begin{array}{ll}
1 \pm \epsilon & \text{if } x[v_1] = \frac{1}{2} \pm \epsilon \lor x[v_2] = \frac{1}{2} \pm \epsilon \\
0 \pm \epsilon & \text{if } x[v_1] = 0 \pm \epsilon \land x[v_2] = 0 \pm \epsilon
\end{array}\right.$
\end{tabular} \\
\hline
$G_{\lnot}(\vert v_1 \vert v)$ & 
\begin{tabular}{c}
$x[v] = \left\{
\begin{array}{ll}
\frac{1}{2} \pm \epsilon & \text{if } x[v_1] = 0 \pm \epsilon \\
0 \pm \epsilon & \text{if } x[v_1] = \frac{1}{2} \pm \epsilon
\end{array}\right.$
\end{tabular} \\
\hline
\end{tabular}
\caption{Alternative Generalized Circuit Constraints \cite{PP21}}
\label{tbl:gcircuit2}
\end{table}

\subsection{Reducing $(\kappa,\delta)$-$\threshold$ to $\epsilon$-HZ}
For any constant $\kappa>0$, \cite{CCPY21} reduced the $\kappa$-$\threshold$ problem with bounded in-degree/out-degree to $(1/n^c)$-HZ where $n$ is the number of goods/agents in the HZ market, which is polynomial in the number of vertices in the threshold game and $1/\kappa$. In contrast, we show that $(\kappa,\delta)$-$\threshold$ with bounded in-degree/out-degree can be reduced to $\epsilon$-HZ where $\epsilon$ is just a constant, thereby showing a stronger hardness of approximation result under Conjecture \ref{conj:gcircuit}. Without loss of generality, we assume the $(\kappa,\delta)$-$\threshold$ problem uses a threshold of $t=1/2$, which is derived from the reduction in the proof of Lemma \ref{lem:eps-delta-threshold-ppad}.

The HZ market we construct is the same as in \cite{CCPY21}. We restate it here for completion. Let $m=\lceil C/\kappa\rceil$, where $C$ is a large constant. Let $M_H$ be an HZ market consisting of the following groups of agents and goods: for each vertex $v\in V$, we construct a group of agents $A_v$ ($5m^{10}$), and three groups of goods $G_{v,1}$ ($m^{10}+S_v$), $G_{v,2}$ ($2m^{10}$), $G_{v,3}$ ($2m^{10}$), where the quantity of each group is indicated in the parentheses and
\begin{align*}
    S_v = (24m^3+12m)\cdot \outdeg(v)+(24m^3+15m)\cdot \indeg(v)-3m;
\end{align*}
For each edge $e\in E$, we construct five groups of agents $A_{e,*}(64m^5)$, $A_{e,1}(48m^3)$, $A_{e,2,\ell}$ (6 for each $\ell\in [m]$), $A_{e,3,\ell}$ (8 for each $\ell\in [m]$), $A_{e,4,\ell}$ (18 for each $\ell\in [2m]$) and a group of goods $G_e$ ($32m^5$).
We also add a group of dummy goods $G_D$ ($3m|V|+(32m^5+23m)|E|$) to make the number of agents and goods the same. Let $A$ be the set of all agents created and $G$ be the set of all goods created. Let $n$ be the total number agents/goods in the market, then $n=|A|=|G|=5m^{10}|V|+(64m^5+48m^3+50m)|E|$. Let $G_v = G_{v,1}\cup G_{v,2}\cup G_{v,3}$, and \[\textstyle
A_e = A_{e,1}
\cup \bigcup_{\ell\in[m]} A_{e,2,\ell}
\cup \bigcup_{\ell\in[m]} A_{e,3,\ell}
\cup \bigcup_{\ell\in[2m]} A_{e,4,\ell}.
\] Then $|G_v|=5m^{10}+S_v$ and $|A_e|=O(m^3)$. 
The utilities of the different groups of agents are as follows:
\begin{align*}
    &A_v: 1 \tfor G_{v,3}, \frac{m^2+1}{4m^2-2} \tfor G_{v,2}, \frac{1}{2m^2-1} \tfor G_{v,1}
    \\
    &A_{e,*}: 1 \tfor G_e\\
    &A_{e,1}: 1 \tfor G_{u,3}, \frac{1}{2}\tfor G_{v,1}\\
    &A_{e,2,\ell}: 1 \tfor G_e, \frac{\ell}{2m^3}\tfor G_{v,1}\\
    &A_{e,3,\ell}: 1 \tfor G_e, \frac{\ell}{2m^3}\tfor G_{u,1}\\
    &A_{e,4,\ell}: 1 \tfor G_e,\frac{\ell}{2m^3}\tfor G_{v,1}, \frac{1}{4}+\frac{1}{4m^2}+\frac{1}{m^3}\tfor G_{u,2}.\\
\end{align*}
Observe that each agent has a favorite good of utility 1, which is either $G_e$ for some edge $e$ or $G_{v,3}$ for some vertex $v$.

Note that $n\leq 6m^{10}|V|$ for $C$ large enough since the threshold graph $H$ has bounded degree and thus $|E|=O(|V|)$. Consider an $\epsilon$-HZ equilibrium $(x,p)$, where $\epsilon$ will be set later. Let $p(G_e)$ and $\bar{p}(G_e)$ denote the minimum and maximum prices of any good in $G_e$ (i.e. identical copies of goods in $G_e$ can have different prices). Similarly we define the minimum and maximum prices for groups of goods $G_{v,1}, G_{v,2}, G_{v,3}$. 

Similar to \cite{CCPY21}, we will rely on reasoning about the individual optimization LP and its dual defined as follows:
\begin{definition}[Individual Optimization LP]
    \begin{align*}
        \max &\sum_{j\in [n]}u_{ij}x_{ij}\\
        s.t. 
        &\sum_{j\in [n]}p_j x_{ij}\leq 1\\
        &\sum_{j\in [n]}x_{ij}=1\\
        & x_{ij}\geq 0, \forall j\in [n]
    \end{align*}
\end{definition}

\begin{definition}[Dual of Individual Optimization LP]
\label{def:dual-LP}
    \begin{align}
        \min \hspace{3pt}&\alpha_i+\mu_i\nonumber\\
        s.t. \hspace{3pt}& \alpha_i p_j+\mu_i\geq u_{ij}, \forall j\in [n] \label{eq:dual-constraint} \\
        & \alpha_i\geq 0 \nonumber
    \end{align}
\end{definition}
Let $\alpha_i, \mu_i$ be an optimal solution to the dual LP of each agent $i$ in $M_H$. Let $val_p(i)$ denote the optimal value of both of these LPs. This is the maximum utility agent $i$ can get subject to the non-relaxed budget of 1 and prices $p$ from the $\epsilon$-HZ equilibrium. 

\begin{definition}[$\delta$-suboptimal goods \citep{CCPY21}]
A good $j$ is called $\delta$-suboptimal for agent $i$ if $\alpha_i p_j+\mu_i\geq u_{ij}+\delta$.
\end{definition}

We state a few observations that are still true in the current setting. 

\begin{lemma}[\cite{CCPY21}, Lemma 3.5]
\label{lem:suboptimal-alloc}
    For agent $i$, the total allocation in $x_i$ to $\delta$-suboptimal goods is at most $2\epsilon/\delta$.
\end{lemma}

\begin{lemma}
\label{lem: pe-pv3-val}
\begin{align*}
    &\mu_i\geq 0\\
    &p(G_e)\geq 2(1-2\epsilon)\\
    &p(G_{v,3})\geq 5/3\\
    &\alpha_i+\mu_i=val_p(i)\leq 0.9, \forall i\in A
\end{align*}
\end{lemma}
\begin{proof}
    Follows from Lemma 3.3, Lemma 3.13 and Lemma 3.14 in \cite{CCPY21} as their proof still applies when $\epsilon$ is a small constant.
\end{proof}

In what follows, we will keep a counter of bad groups of goods whose minimum prices or allocation does not behave nicely in equilibrium: either the price is too high or a lot of the goods from that group are allocated to agents with 0 utility. We will denote such groups of goods by $B_G\subset V\cup E$ (the vertex goods indexed by $v$ and edge goods indexed by $e$. (Note that we do not differentiate between the groups of goods $G_{v,1}$, $G_{v,2}$, and $G_{v,3}$ for the same vertex nor the groups of agents for the same edge.)

\begin{definition}
We say an agent is \emph{bad} if it has positive utility for any vertex goods or edge goods in $B_G$.
\end{definition}

We first show that there cannot be too many vertex goods $G_{v,3}$ or edge goods $G_e$ with a high price. We focus on these two groups of goods because every agent's favorite good is either $G_{v,3}$ or $G_e$. In what follows, we set the price ceiling to be $f(m)=m^{20}$ and add any groups of goods exceeding this ceiling to the bad goods. We also set $\epsilon = 1/f(m)^2 = 1/m^{40}$.
\begin{lemma}
\label{lem: frac-bad-goods-1}
Let $k_v$ be the fraction of vertices $v$ such that $p(G_{v,3})\geq f(m)$. Let $k_e$ be the fraction of edges $e$ such that $p(G_e)\geq f(m)$. If $f(m)=m^{20}$, then $k_v,k_e\leq O\left(\frac{m^5}{f(m)}\right)=O\left(\frac{1}{m^{15}}\right)$.
\end{lemma}

\begin{proof}
 Since every agent can spend at most $(1+\epsilon)$, and there are $2m^{10}$ goods in $G_{v,3}$ for every $v$ and $32m^5$ goods in $G_e$ for every $e$, we have
    \begin{align*}
        &2m^{10}|V|k_vf(m)\leq n(1+\epsilon) \leq 6m^{10}|V|(1+\epsilon)\\
        &32m^{5}|E|k_ef(m)\leq n(1+\epsilon)\leq 6m^{10}|E|(1+\epsilon).\\
    \end{align*}
The result then follows since $\epsilon<1$. 
\end{proof}

Let $B^1_G$ denote the set of vertices $v$ of or edges $e$ such that $p(G_{v,3})\geq f(m)$ or $p(G_{e})\geq f(m)$. Then any good labeled with vertices or edges in $B^{1}_G$ is bad. Let $B^{1}_A$ be bad agents who have positive utility for any good in $B^1_G$.  

\begin{lemma}
\label{lem: frac bad agents-1 and alpha_i}
\begin{align*}
    &|B_{A}^{1}| = O\left(\frac{m^5}{f(m)}n\right); \\
    & \text{For } i\in A\setminus B_{A}^{1}, \alpha_i=\Omega\left(\frac{1}{f(m)}\right)
\end{align*}
\end{lemma}
\begin{proof}
   To calculate the number of bad agents, we have
   \begin{align*}
    |B^{1}_A|&\leq 5m^{10}|V|k_v+(64m^5+O(m^3))|E|k_e\\
    &\leq (5m^{10}|V|+(64m^5+O(m^3))|E|)O\left(\frac{m^5}{f(m)}\right) = O\left(\frac{m^5}{f(m)}n\right).
    \end{align*}
    For $i\in A\setminus B_{A}^{1}$, without loss of generality, we can assume $i$'s favorite good is $G_e$. Since any agent gets utility 1 from their favorite item, we have
    \begin{align*}
        &\alpha_i p(G_e) +\mu_i\geq 1\\
        & \alpha_i + \mu_i = val_p(i)\\
        \implies &\alpha_i \geq \frac{1-val_p(i)}{p(G_e)-1}\geq \frac{0.1}{f(m)-1} = \Omega\left(\frac{1}{f(m)}\right)
    \end{align*}
    where we use the fact that $val_p(i)\leq 0.9$ and $p(G_e), p(G_{v,3})>1$ from Lemma \ref{lem: pe-pv3-val}.
\end{proof}

Let $g_i$ denote the amount of budget that an agent wastes with respect to the non-relaxed budget of 1, i.e. not spend or spend on goods with 0-utility or on goods with a higher price than the minimum price of its group. More formally, we say
\begin{align*}
    g_i = 1-\sum_{j:u_{ij}>0} x_{ij}p(G_j)
\end{align*}
where $G_j$ denote the group that a good $j$ belongs to and $p(G_j)$ is the minimum price of this group. Note that $g_i$ might be negative since each agent might spend up to $1+\epsilon$.

We show that $g_i$ is subconstant in $m$ with the right choice of $\epsilon$ for any agent not in $B_{A}^{1}$. The intuition is that any such agent can always purchase their favorite item with a bounded price and they do not already have their favorite item in full capacity (since their favorite item is under-supplied), implying their allocation would not be approximately optimal if they waste a lot of their budget. 

\begin{lemma}
\label{lem:g(1/m)}
   For $i\in A\setminus B_{A}^{1}$, 
   \begin{equation*}
       g_i\leq g(1/m)
   \end{equation*}
   where $g(1/m) = O(1/f(m))$.
\end{lemma}

\begin{proof}
    Since $\alpha_i+\mu_i\leq 0.9$ for all agent $i$ by Lemma \ref{lem: pe-pv3-val}, we know that each agent is allocated $\leq 0.9$ unit of their favorite item. Thus by unit-demand, they must be allocated $\geq 0.1$ unit of items that give them utility at most 1/2 since any agent's second favorite item provides them with utility at most 1/2. With a budget of $g_i\leq 1$, they can swap $g_i/f(m)<0.1$ unit of their current allocation with their favorite good for a per-unit gain of $1/2$ in utility. Since $(x,p)$ is $\epsilon$-HZ, we have
    \begin{equation*}
        \frac{g_i}{f(m)}\frac{1}{2}\leq \epsilon\implies  g_i\leq 2\epsilon f(m)=O\left(\frac{1}{f(m)}\right)
    \end{equation*}
    since $\epsilon = \frac{1}{f(m)^2}$.
\end{proof}

\begin{corollary}
\label{cor:spend-in-min-price}
$i\in A\setminus B_{A}^{1}$, 
\begin{align*}
    \sum_{j:u_{ij}>0} x_{ij}p(G_j) = 1\pm O(1/m^{20})
\end{align*}
\end{corollary}
\begin{proof}
    Follows from budget constraint $1+\epsilon=1+1/m^{40}$ and Lemma \ref{lem:g(1/m)}.
\end{proof}

Let $x^{0}(G_j, A_k)$ denote the total allocation of goods in $G_j$ to agents in $A_k$ such that the allocation gives 0 utility to the agents. Since we know $p(G_e)$ is bounded from below, we can bound the sum of allocations of $G_e$ to 0-utility agents over all edges $e$.
\begin{lemma}
\begin{equation*}
    \bigcup_{e\in E} x^{0}(G_e, A) = O\left(\frac{m^5}{f(m)}\right)n.
\end{equation*}
\end{lemma}
\begin{proof}
    \begin{align*}
        \bigcup_{e\in E}x^{0}(G_e,A) &\leq 
        \bigcup_{e\in E}x^{0}(G_e,B_A^{1}) + \bigcup_{e\in E}x^{0}(G_e,A\setminus B_A^{1})\\
        &\leq |B_{A}^{1}|+ (|A|-|B_{A}^{1}|)\frac{g(1/m)}{2-O(\epsilon)}\\
        &= O\left(\frac{m^5}{f(m)}\right)n + \left(1-O\left(\frac{m^5}{f(m)}\right)\right)n\frac{O(1/f(m))}{2-O(\epsilon)}\\
        &= O\left(\frac{m^5}{f(m)}\right)n
    \end{align*}
    since $\epsilon = 1/f(m)^2$. The second inequality follows from Lemma \ref{lem:g(1/m)} and Lemma \ref{lem: pe-pv3-val}.
\end{proof}

Similarly, we can bound the allocation of $G_v$ goods to 0-utility agents if the price of the goods are bounded from below.

\begin{lemma} 
Fix any $\ell\in [3]$, then
    \begin{equation*}
     \bigcup_{v: p(G_{v,\ell})\geq 1/m^4} x^{0}(G_{v,\ell},A)= O\left(\frac{m^5}{f(m)}\right)n.
    \end{equation*}
\end{lemma}
\begin{proof}
    \begin{align*}
        \bigcup_{v: p(G_{v,\ell})\geq 1/m^4}x^{0}(G_{v,\ell},A) &\leq 
        \bigcup_{v: p(G_{v,\ell})\geq 1/m^4}x^{0}(G_{v,\ell},B_A^{1}) + \bigcup_{v: p(G_{v,\ell})\geq 1/m^4}x^{0}(G_{v,\ell},A\setminus B_A^{1})\\
        &\leq |B_{A}^{1}|+ (|A|-|B_{A}^{1}|)\frac{g(1/m)}{1/m^4}\\
        &= O\left(\frac{m^5}{f(m)}\right)n + \left(1-O\left(\frac{m^5}{f(m)}\right)\right)n\frac{O(1/f(m))}{1/m^4}\\
        &= O\left(\frac{m^5}{f(m)}\right)n,
    \end{align*}
    where the second inequality follows from Lemma \ref{lem:g(1/m)}.
\end{proof}

\begin{lemma}
    Let $k_{e}^{0}$ be the fraction of edges $e$ where $x^{0}(G_e,A)>1/m$. Then $k_{e}^{0}\leq O\left(\frac{m^{16}}{f(m)}\right)=O\left(\frac{1}{m^4}\right)$.
\end{lemma}

\begin{proof}
    \begin{align*}
        &\bigcup_{e\in E}x^{0}(G_e, A) = O\left(\frac{m^5}{f(m)}\right)n\leq O\left(\frac{m^5}{f(m)}\right)6m^{10}|E|\\
       & \bigcup_{e\in E}x^{0}(G_e, A)> \frac{1}{m}|E|k_{e}^{0}\\
       &\implies k_{e}^{0} \leq O\left(\frac{m^{16}}{f(m)}\right)
    \end{align*}
\end{proof}
Let $B^2_G$ denote the set of edges $e$ such that $x^{0}(G_e,A)>1/m$.

\begin{lemma}
    For any $\ell\in [3]$, let $k_{v,\ell}^{0}$ be the fraction of vertices $v$ where $p(G_{v,\ell})\geq 1/m^4$ and $x^{0}(G_{v,\ell},A)>1/m$. Then $k_{v,\ell}^{0}\leq O\left(\frac{m^{16}}{f(m)}\right)=O\left(\frac{1}{m^4}\right)$.
\end{lemma}

\begin{proof}
    \begin{align*}
       &\bigcup_{v: p(G_{v,\ell})\geq 1/m^4} x^{0}(G_{v,\ell},A)= O\left(\frac{m^5}{f(m)}\right)n\leq O\left(\frac{m^5}{f(m)}\right)6m^{10}|V|\\
       & \bigcup_{v: p(G_{v,\ell})\geq 1/m^4} x^{0}(G_{v,\ell},A)> \frac{1}{m}|V|k_{v,\ell}^{0}\\
       &\implies k_{v,\ell}^{0} \leq O\left(\frac{m^{16}}{f(m)}\right)
    \end{align*}
\end{proof}

Let $B^3_G$ denote the set of vertices $v$ such that $p(G_{v,\ell})\geq 1/m^4$ and $x^{0}(G_{v,\ell},A)>1/m$ for some $\ell\in [3]$. 

The upper bound on allocation of goods to 0-utility agents makes sure that the price of groups of goods behave like they do in a local market. Let $x^{+}(G_j, A_k)$ denote the total amount of goods in group $G_j$ that is assigned to agents in group $A_k$ such that the assignment gives positive utility to the agents.
\begin{lemma}
\label{lem:ge-price-eps-delta}
For $e\notin B_G^{2}$, $p(G_e) = 2\pm O(1/m^4)$.
\end{lemma}
\begin{proof}
    Since $e\notin B_G^2$, we have $x^{0}(G_e,A)\leq 1/m$. Then $x^{+}(G_e,A) = x(G_e,A_e\setminus A_{e,1})\geq 32m^5-1/m$. By budget constraint, we have
    \begin{align*}
        p(G_e)\leq \frac{(64m^5+50m)(1+\epsilon)}{32m^5-1/m} = 2+O(1/m^4)
    \end{align*}
    The claim then follows from Lemma \ref{lem: pe-pv3-val}.
\end{proof}

Let $y_{\ell}(v) = x(G_{v,\ell}, A_v)$ for any $\ell\in [3]$. Let $q_{\ell}(v) = p(G_{v,\ell})$. We might use $y_{\ell}$ and $q_{\ell}$ when $v$ is clear from context. Let $u_{\ell}$ denote the utility $A_v$ has for $G_{v,\ell}$. We rederive a few lemmas regarding the allocation and prices of vertex goods for a subset of vertices.
\begin{lemma}
\label{lem:y1y2y3-eps-delta}
    For $v\notin B_G^3$, we have $y_1\geq m^{10}-O(m^3)$, $y_2,y_3\geq 2m^{10}-O(m^3)$.
\end{lemma}
\begin{proof} (Adapted from proof of Claim 3.17 in \cite{CCPY21}). Since $v\notin B_G^3$, we have $q_\ell<1/m^4$ or $x^{0}(G_{v,\ell},A)\leq 1/m$ for every $\ell\in [3]$.
 Let $\alpha,\mu$ be an optimal solution to the dual LP for agents in $A_v$.
\begin{itemize}
    \item Case 1: $\mu\geq u_1/2=\Omega(1/m^2)$. Then any good that $A_v$ has 0 utility for is $\mu$-suboptimal for them. Thus the total allocation of $A_v$ to $\overline{G_v}$ (i.e. not $G_v$ goods) is bounded by $5m^{10} * \frac{2\epsilon}{\Omega(1/m^2)} = O(m^{12})\epsilon<1$ by Lemma \ref{lem:suboptimal-alloc} and $\epsilon=1/m^{40}$. Thus $y_1+y_2+y_3\geq 5m^{10}-1$. Since $y_1\leq |G_{v,1}|=m^{10}+S_v=m^{10}+\Theta(m^3)$ and $y_2,y_3\leq 2m^{10}$, we have $y_1\geq m^{10}-1$ and $y_2,y_3\geq 2m^{10}-O(m^3)$.
    \item Case 2: $\mu < u_1/2$. For $\ell\in [3]$, since $\alpha q_\ell+\mu\geq u_\ell\geq u_1$, we have $\alpha q_\ell\geq u_1/2$. By Lemma \ref{lem: pe-pv3-val}, $\alpha<1$, thus $q_\ell\geq \frac{u_1}{2\alpha}\geq \frac{u_1}{2}=\Omega(1/m^2)>1/m^4$ for all $\ell\in [3]$. Thus by $v\notin B_G^3$, we have $x^{0}(G_{v,\ell},A)\leq 1/m$ for all $\ell\in [3]$, which implies $x^{0}(G_{v},A)\leq O(1/m)$. Since
    $G_{v}$ only gives positive utility to $A_v$ agents and $A_e$ agents where $e=(u,v)$ or $e=(v,k)$ for some $u$ and $k$, and there are $\Theta(m^3)$ such $A_e$ agents (the in-degree and out-degree of $H$ are bounded), the total allocation of $G_v$ to $A_v$ is $y_1+y_2+y_3\geq 5m^{10}+S_v-O(m^3)-O(1/m)\geq 5m^{10}-O(m^3)$. Thus $y_1\geq m^{10}-O(m^3)$ and $y_2,y_3\geq 2m^{10}-O(m^3)$.
\end{itemize}
\end{proof}

\begin{lemma}
\label{lem:gvl-prices-eps-delta}
    For all $v\notin B_{G}^1\cup B_{G}^3$, 
    \begin{align}
        & 0\leq p(G_{v,1})\leq \frac{1}{m^2} + O\left(\frac{1}{m^6}\right)\\
        & p(G_{v,2})=\frac{1+p(G_{v,1})}{2}\pm O\left(\frac{1}{m^7}\right)\\
        & p(G_{v,3})=2-p(G_{v,1})\pm O\left(\frac{1}{m^7}\right)
    \end{align}
    and for any $\ell\in [3]$, $G_{v,\ell}$ is at most $\delta$-suboptimal for any agent in $A_v$ for $\delta=20\epsilon$.
\end{lemma}

\begin{proof} (Adapted from proof of Lemma 3.10 in \cite{CCPY21}).
Fix $v\notin B_{G}^1\cup B_{G}^3$.
Let $\alpha,\mu$ be an optimal solution to the dual LP for agents in $A_v$. Since $A_v\subset A\setminus B_A^1$, $\alpha\geq \Omega(1/f(m))$ by Lemma \ref{lem: frac bad agents-1 and alpha_i}. Since $v\notin B_{G}^3$, Lemma \ref{lem:y1y2y3-eps-delta} applies.

First we note that if for any $\ell\in [3]$, $G_{v,\ell}$ is $\delta$-suboptimal for $A_v$, then we can assign at most $2\epsilon/\delta$ of $G_{v,\ell}$ per agent in $A_v$ by Lemma \ref{lem:suboptimal-alloc}. Since we have to assign $y_\ell\geq m^{10}-O(m^3)$ in total by Lemma \ref{lem:y1y2y3-eps-delta}, we need
    \begin{align*}
    &5m^{10}2\epsilon/\delta\geq m^{10}-O(m^3)\implies \delta\leq 10\epsilon+O(\epsilon/m^7)<20\epsilon.
    \end{align*}
    for $m$ large enough.
    Thus $G_{v,\ell}$ cannot be $20\epsilon$-suboptimal for $A_v$ for any $\ell\in [3]$. Let $\delta=20\epsilon$. This implies
    \begin{equation}
    \label{ineq: constraint-bdd-eps-delta}
        u_{\ell}\leq \alpha q_\ell +\mu \leq u_\ell+\delta, \forall \ell\in [3].
    \end{equation}
    Since utilities are designed so that $u_2=(3u_1+u_3)/4$, we have
    \begin{align*}
        &\alpha\left(\frac{3q_1+q_3}{4}\right)+\mu-\delta\leq u_2\leq \alpha q_2+\mu\leq u_2+\delta\leq \alpha\left(\frac{3q_1+q_3}{4}\right)+\mu+\delta\\
        \implies & q_2=\left(\frac{3q_1+q_3}{4}\right) \pm O\left(\frac{\delta}{\alpha}\right)
        =\left(\frac{3q_1+q_3}{4}\right) \pm O\left(\frac{1}{m^9}\right)
    \end{align*}
    The last line uses the fact that $O(\delta/\alpha)=O(\epsilon f(m)) = O(1/m^9)$ since $\alpha\geq \Omega(1/f(m))$, $\epsilon = 1/f(m)^2$ and $f(m)=m^{20}$.
    Now, we can use the total payment of agents in $A_v$ to bound $q_1,q_2, q_3$.  Since $A_v\subset A\setminus B_A^1$, by Corollary \ref{cor:spend-in-min-price} we have
\begin{align*}
\label{eq: good-Av-payment}
&q_1y_1 + q_2y_2 + q_3y_3 = 5m^{10}(1\pm O(1/m^{20}))\\
&\implies q_1y_1+\left(\left(\frac{3q_1+q_3}{4}\right) \pm O(1/m^9)\right)y_2 + q_3y_3 = 5m^{10}(1\pm O(1/m^{20}))
\end{align*}
where $f(m)=m^{20}$.
    
By Lemma \ref{lem:y1y2y3-eps-delta} we conclude that
    \begin{align}
        &q_2=\frac{1+q_1}{2}\pm O\left(\frac{1}{m^7}\right)\\
        & q_3=2-q_1\pm O\left(\frac{1}{m^7}\right)\label{eqn:q1q3-eps-delta}.
    \end{align}
    Now it remains to bound $q_1$. We first observe that $q_1<q_3$ since if not, then we have 
    \begin{align*}
        \alpha q_1+\mu\geq  \alpha q_3+\mu\geq u_3=1.
    \end{align*}
    Since $u_1 = \Theta(1/m^2)$, $G_{v,1}$ is $(1-\Theta(1/m^2))$-suboptimal for $A_v$, a contradiction to $\delta = 20\epsilon$. Thus $q_1< q_3$. By \ref{ineq: constraint-bdd-eps-delta}, since $q_1,q_3\geq 0$, we have 
    \begin{align*}
    &q_3*u_{1}\leq q_3(\alpha q_1 +\mu) \leq q_3(u_1+\delta) \\
    &q_1*u_{3}\leq q_1(\alpha q_3 +\mu) \leq q_1(u_3+\delta)\\
        &\implies\mu\leq \frac{u_1q_3-u_3q_1+q_3\delta}{q_3-q_1} =\frac{u_1q_3-u_3q_1+O(\delta)}{q_3-q_1}.
    \end{align*}
    since $q_3\leq 3$ by \ref{eqn:q1q3-eps-delta}.
    Note that the denominator is positive.
    Suppose $q_1\geq 1/m^2+1/m^6$, then $q_3\leq 2-1/m^2$ by \ref{eqn:q1q3-eps-delta}. Then $u_1q_3-u_3q_1\leq -1/m^6$. Since $\delta=20\epsilon=O(1/m^{40})$, the numerator is negative. Thus $\mu<0$, a contradiction to $\mu\geq 0$ by Lemma \ref{lem: pe-pv3-val}. Thus
    \begin{equation*}
        q_1<1/m^2+1/m^6.
    \end{equation*}
\end{proof}

\begin{lemma}
\label{lem:x+vneqv-eps-delta}
For any $v\notin B_G^1\cup B_G^3$,
    \begin{enumerate}
        \item If $x^{+}(G_v, \overline{A_v})\geq S_v+1$, then $p(G_{v,1})=1/m^2\pm O(1/m^9)$
        \item If $x^{+}(G_v, \overline{A_v})\leq S_v-1$, then $p(G_{v,1})\leq \kappa/m^2$.
    \end{enumerate}
\end{lemma}
\begin{proof} (Adapted from proof of Lemma 3.11 in \cite{CCPY21}).

\noindent\begin{itemize}
        \item Case 1: $x^{+}(G_v, \overline{A_v})\geq S_v+1$. Since there are $5m^{10}+S_v$ goods in $G_v$, we can allocate at most $5m^{10}-1$ goods to $A_v$, so there exists an agent in $A_v$ who is allocated at least $1/5m^{10}$ units of goods outside of $G_v$. Let $\alpha, \mu$ be an optimal solution to the dual of LP of this agent. Since this agent has 0 utility for these goods, they are $\mu$-suboptimal, which means $1/5m^{10}\leq 2\epsilon/\mu$, so $\mu=O(m^{10}\epsilon)$. By Lemma \ref{lem:gvl-prices-eps-delta}, \ref{ineq: constraint-bdd-eps-delta} and \ref{eqn:q1q3-eps-delta}, we have 
        \begin{align}
            &\alpha(q_1+q_3)+2\mu =u_1+u_3+ O(\delta)\nonumber\\
            \implies & \alpha = \frac{u_1+u_3+O(\epsilon)-2*O(m^{10}\epsilon)}{q_1+q_3}\nonumber\\
             \implies & \alpha = \frac{u_1+u_3}{2}(1\pm O(1/m^7))\geq \frac{1}{2}(1\pm O(1/m^7))\label{eq:case1-eps-delta}
        \end{align}
        since $\delta=20\epsilon$ and $\epsilon = 1/m^{40}$. Again by \ref{ineq: constraint-bdd-eps-delta}, we have 
        \begin{align*}
            q_1 &=\frac{u_1-\mu+O(\delta)}{\alpha}\\
            &=\frac{u_1}{\alpha}+\frac{-O(m^{10}\epsilon)+O(\epsilon)}{\alpha}\\
            &=\frac{u_1}{\alpha}\pm O(1/m^{10})\text{  by \ref{eq:case1-eps-delta}}\\
            &=\frac{u_1}{\frac{u_1+u_3}{2}(1\pm O(1/m^7))}\pm O(1/m^{10})\text{  by \ref{eq:case1-eps-delta}}\\
             &=1/m^2\pm O(1/m^9).
        \end{align*}
        
        \item Case 2: $x^{+}(G_v, \overline{A_v})\leq S_v-1$. 
        Then
        \begin{align*}
            x^0(G_v,\overline{A_v})+x(G_v, A_v)\geq 5m^{10}+1
            \implies x^0(G_v,\overline{A_v})\geq 1
        \end{align*}
        since $x(G_v, A_v)\leq |A_v|=5m^{10}$.
        This implies $x^0(G_{v,1},\overline{A_v})+x^0(G_{v,2},\overline{A_v})+x^0(G_{v,3},\overline{A_v})\geq 1$. Since $v\notin B_G^3$, it must be the case that either $p(G_{v,\ell})<1/m^4$ or $x^0(G_{v,\ell}, A)\leq 1/m$ for all $\ell\in [3]$. Since $p(G_{v,\ell})>1/m^4$ for $\ell=2,3$ by Lemma \ref{lem:gvl-prices-eps-delta}, it must be the case that $$p(G_{v,1})\leq \frac{1}{m^4}\leq  \frac{1}{m^3}\frac{\kappa}{C}\leq \frac{\kappa}{m^2}$$ where $m=\lceil \frac{C}{\kappa}\rceil$. 
    \end{itemize}
\end{proof}

\begin{lemma}
\label{lem:x+ueve-eps-delta}
Consider any $e=(u,v)\in E$ such that $e\notin B_{G}^1\cup B_G^2$ and $u,v\notin B_{G}^1\cup B_G^3$, then
\begin{align*}
    &x^{+}(G_u, A_e)=24m^3+12m\pm O(1)\\
    &x^{+}(G_v, A_e)=24m^3+15m-6m^3p(G_{u,1})\pm O(1)\\
\end{align*}
\end{lemma}

\begin{proof}
The proof of Lemma \ref{lem:x+ueve-eps-delta} is nearly identical to the proof of Lemma 3.12 in \cite{CCPY21}, so we only highlight the modifications that need to be made.
At a high level, the proof of Lemma 3.12 in \cite{CCPY21} characterizes the optimal dual solution $(\alpha, \mu)$ for each agent and then applies their bounds on the allocations to each agent and prices. 
Thus, the only modifications that need to be made are to use our bounds on the allocations and prices when they use theirs.
We provide a mapping from their results to ours below so that this substitution is easier.
\begin{center}
\begin{tabular}{ c|c } 
    \cite{CCPY21} & Here \\ 
    \hline
    Lemma 3.5 & Lemma \ref{lem:suboptimal-alloc} \\
    Corollaries 3.6 + 3.15 & Lemma \ref{lem:suboptimal-alloc} + Corollary \ref{cor:spend-in-min-price} (requires $e\notin B_G^1$) \\ 
    Lemma 3.10 & Lemma \ref{lem:gvl-prices-eps-delta} (requires $u,v\notin B_G^1\cup B_G^3$) \\
    Lemma 3.16 & Lemma \ref{lem:ge-price-eps-delta} (requires $e\notin B_G^2$)
\end{tabular}
\end{center}
A careful reader may realize that whenever \cite{CCPY21} obtains a bound up to error $O(m^6\varepsilon)$, this error is $O(1/\mathrm{poly}(n))$ for them since $\varepsilon= O(1/\mathrm{poly}(n))$ in their setting.
Meanwhile, for us, $O(m^6\varepsilon) = O(1/m^{34})$ since we chose $\varepsilon = O(1/m^{40})$.
Nonetheless, this additional amount of error does not quantitatively affect the proof.

\end{proof}

\begin{lemma}
\label{lem:bound on bad fraction}
    Let $\Delta$ be the fraction of $v\in V$ s.t. $v\notin B_G^1\cup B_G^3$ and for any edge $(u,v)\in E$ and $(v,w)\in E$, we have $u,w\notin B_G^1\cup B_G^3$ and $(u,v),(v,w)\notin B_G^1\cup B_G^2$. Then $\Delta\geq 1-O(1/m^4)$.
\end{lemma}
\begin{proof}
    By definition of $B_G^1$, $B_G^2$, and $B_G^3$, we have 
    \begin{align*}
        &|B_G^1|\leq k_v|V|+k_e|E|\\
        &|B_G^2|\leq k_e^0|E|\\
        &|B_G^3|\leq k_{v,\ell}^0*3|V|\\ 
    \end{align*}
    Thus 
    \begin{align*}
        &\Delta|V|\geq |V|-(|B_G^1|+|B_G^3|)-6(|B_G^1|+|B_G^3|)-6(|B_G^1|+|B_G^2|)\\
        &\geq |V|-O(|B_G^1|)-O(|B_G^2|)-O(|B_G^3|)\\
        &=|V|-|V|O(k_v+k_e+k_e^0+k_{v,\ell}^0)\\
        &= |V|-|V|(O(1/m^{15})+O(1/m^{15})+O(1/m^4)+O(1/m^4)) \\
        &= |V|(1-O(1/m^4))
    \end{align*}
    where we use the fact that $|E|=O(|V|)$. Here we first remove the bad vertices from $B_G^1$ and $B_G^3$ themselves and then we remove any vertex destroyed by any of its bad neighbors (each bad vertex can destroy up to a constant number of vertices since the degree of $H$ is bounded) and finally remove any vertex destroyed by any of its incident bad edges (each bad edge can destroy two vertices). Here we crucially rely on the fact that $H$ has bounded in-degree and out-degree.
\end{proof}

Now we are ready to show the core Lemma.  
\begin{lemma}
\label{lem: reduce-threshold-hz}
    There is a polynomial-time reduction from $(\kappa,\delta)$-$\threshold$ game to $\epsilon$-HZ.
\end{lemma}

\begin{proof}(Adapted from proof of Theorem 3.1 in \cite{CCPY21}.)

The proof relies on Lemma \ref{lem:gvl-prices-eps-delta} (requires $u,v\notin B_G^1\cup B_G^3$), Lemma \ref{lem:x+vneqv-eps-delta} (requires $v\notin B_G^1\cup B_G^3$), and Lemma \ref{lem:x+ueve-eps-delta} (requires $e\notin B_G^1\cup B_G^2$).

For any instance of the threshold game $H=(V,E)$ where the in-degree and out-degree are bounded and the threshold is $1/2$, we can create the HZ market $M_H$ as described before. Let $(x,p)$ be an $\epsilon$-HZ of $M_H$ where $\epsilon=\frac{1}{m^{40}} = \frac{1}{{\lceil C/\kappa\rceil}^{40}}$ for some large constant $C$. Let $x_v=\min (1,m^2p(G_{v,1}))\in [0,1]$. Then we claim that $x=(x_v)_{v\in V}$ is a $(\kappa,\delta)$-approximate equilibrium of $H$. 

Consider only the fraction of $v\in V$ s.t. $v\notin B_G^1\cup B_G^3$ and for any edge $(u,v)\in E$ and $(v,w)\in E$, we have $u,w\notin B_G^1\cup B_G^3$ and $(u,v),(v,w)\notin B_G^1\cup B_G^2$. Denote this set of vertices by $S$. Then by Lemma \ref{lem:bound on bad fraction}, we only removed at most $O(1/m^4)=c'/m^4$ fraction of the vertices for some constant $c'\geq 1$. Since $m=\lceil C/\kappa\rceil$ for some constant $C$. We can choose $C>c'/\delta$ so that $c'/m^4 \leq \kappa^4\delta^4/c'^3<\delta$ since $0<\kappa,\delta<1$. Thus for any $v\in S$, Lemma \ref{lem:gvl-prices-eps-delta}, Lemma \ref{lem:x+vneqv-eps-delta}, and Lemma \ref{lem:x+ueve-eps-delta} apply. It remains to be shown that for any vertex $v\in S$, the threshold condition holds.
\begin{itemize}
\item Case 1: $\sum_{u\in N_v}x_u>0.5+\kappa$. By definition of $x_u$, we have $x_u\leq m^2 p(G_{u,1})$. This implies $ p(G_{u,1})\geq x_u/m^2\implies \sum_{u\in N_v}p(G_{u,1})\geq \sum_{u\in N_v}x_u/m^2>(1/m^2)(0.5+\kappa)$. Then by Lemma \ref{lem:x+ueve-eps-delta}, we have
\begin{align*}
    x^{+} (G_v, \overline{A_v}) &= \sum_{v\in e} x^{+}(G_v, A_e)\\
    &\leq \outdeg(v)\cdot (24m^3+12m+ O(1)) +\sum_{u\in N_v} (24m^3+15m-6m^3p(G_{u,1})+ O(1))\\
    &= S_v+3m-6m^3\sum_{u\in N_v}p(G_{u,1})+O(1)\\
    &\leq S_v+3m-6m^3(1/m^2)(0.5+\kappa)+O(1)\\
    &= S_v-6m\kappa + O(1)
    <S_v-1.
\end{align*}
since $m=\lceil C/\kappa\rceil$ and we can choose the $C$ to be larger than the $O(1)$ term.
Then by Lemma \ref{lem:x+vneqv-eps-delta}, we have $p(G_{v,1})\leq \kappa/m^2$, which implies $x_{v}\leq \kappa$, as desired.
\item Case 2: $\sum_{u\in N_v}x_u<0.5-\kappa$. By Lemma \ref{lem:gvl-prices-eps-delta}, we have $p(G_{u,1})\leq 1/m^2+O(1/m^6)$.
 Thus 
\begin{align*}
&m^2p(G_{u,1})\leq 1+O(1/m^4)\\
    &\implies x_u = \min(1,m^2p(G_{u,1})) \geq m^2p(G_{u,1})-O(1/m^4)\\
    &\implies p(G_{u,1})\leq x_u/m^2+O(1/m^6)\\
    &\implies \sum_{u\in N_v}p(G_{u,1})\leq \sum_{u\in N_v}x_u/m^2+O(1/m^6)<(1/m^2)(0.5-\kappa+O(1/m^4))
\end{align*}
by bounded in-degree.
Then by Lemma \ref{lem:x+ueve-eps-delta}, we have
\begin{align*}
    x^{+}(G_v, \overline{A_v})&\geq \outdeg(v)\cdot (24m^3+12m- O(1)) +\sum_{u\in N_v} (24m^3+15m-6m^3p(G_{u,1})-O(1))\\
    &= S_v+3m-6m^3\sum_{u\in N_v}p(G_{u,1})-O(1)\\
    &\geq S_v+3m-6m^3(1/m^2)(0.5-\kappa+O(1/m^4))-O(1)\\
    &= S_v+6m\kappa -O(1/m^3)- O(1)
    > S_v+1.
\end{align*}
Thus by Lemma \ref{lem:x+vneqv-eps-delta}, $p(G_{v,1})\geq 1/m^2-O(1/m^9)$, which means $x_v\geq 1-O(1/m^7)\geq 1-\kappa$, as desired.
\end{itemize}
\end{proof}

\begin{proof}
    Theorem \ref{thm:conj-implies-eps-hz-ppad} follows from Lemma \ref{lem:eps-delta-threshold-ppad} and Lemma \ref{lem: reduce-threshold-hz}.
\end{proof}

\section*{Acknowledgments}
The authors would like to thank Aviad Rubinstein of Stanford University for discussions about the $(\epsilon,\delta)$-$\gcircuit$ problem. This research is in part supported by National Science Foundation Alan T. Waterman award grant 1933331.

\bibliographystyle{ACM-Reference-Format}
\bibliography{references}

\appendix
\section{Reduction of $(\epsilon,\delta)$-$\gcircuit$ Problem to Constant Fan-Out}
\label{append:const fan-out}

By definition of a generalized circuit, each gate has a fan-in of at most 2, i.e. it takes inputs from at most two nodes in the circuit. However, each gate can have an unbounded fan-out, i.e. the output node of the gate can have a wire connecting to many gates in the circuit, serving as one of their input values. We will show below that it is without loss of generality to consider only the $(\epsilon,\delta)$-$\gcircuit$ problems with a constant fan-out.

\begin{proof}
Given an instance of the $(\epsilon,\delta)$-$\gcircuit$ Problem $S = (V,\mathcal{T})$ for some constant $\epsilon, \delta>0$, we show a reduction to an instance of $(\epsilon',\delta')$-$\gcircuit$ Problem $S'$ with constant fan-out $10/\delta$ and $\epsilon'=\epsilon, \delta' = \delta/10$. Let $n=|\mathcal{T}|$ be the number of gates in $S$. Consider the set of gates $L$ that have a fan-out $\geq 10/\delta$. Since each gate in $S$ has a fan-in of at most 2, and each wire either connects a node to a gate or a gate to its unique output node, there are at most $3n$ wires in $S$ and at most $2n$ wires contribute to the fan-in. Thus there are at most $2n/(10/\delta)=\delta n/5$ gates in $L$. We construct $S'$ as follows: copy the gates and wires in $S$. For each gate $T\in L$, let $a$ be the output node of $T$ and let $T_1,\dots, T_k$ denote the gates that the output of $T$ provides inputs into. We remove all the wires connecting the output node of $T$ to each gate $T_i$ for $i\in [k]$ and add a constant 0 gate that provides input to these $T_i$ gates (Note: 0 is a valid input to any gate). More formally, for each $T_i$ gate, we construct a constant $G(0||a_i)$ gate whose output node $a_i$ should satisfy $x[a_i]=0\pm \epsilon$. Then we modify $T_i = G(|a,b|c)$ to $T_i' = G(|a_i,b|c)$. Now the new instance $S'$ has fan-out at most $10/\delta$. Consider an $(\epsilon'=\epsilon,\delta'=\delta/10)$-approximate solution $x'$ to the new instance $S'$ (an $\epsilon$-approximate solution exists to any generalized circuit instance even without the $\delta$ slack). Then we can construct a solution $x$ to $S$ as follows: for each gate $T\in L$ whose output node is $a$, we let $x[a]=0$. If $x'[a]\neq 0$, we declare the gate $T$ corrupted. Now, for each gate $T_1,\dots, T_k$ that gets input from $T$, we check if the $x'[a_i]=0$. If not, we declare $T_i$ corrupted. For any gate not in $L$ and does not receive input from a gate in $L$, it is corrupted in $S$ if and only if it is corrupted in $S'$ since they get the same input values and output values. Since $x'$ is an $(\epsilon, \delta/10)$-approximate solution to $S'$, and we added at most $n$ gates, there are at most $2n(\delta/10)=\delta n/5$ corrupted gates in $S'$. Thus the total number of corrupted gates in $S$ induced by $x$ is at most $\delta n/5 + |L|\leq 2\delta n/5<\delta n$. For any uncorrupted gates, it satisfies the gate to an additive error of $\epsilon$. Thus $x$ is indeed a $(\epsilon,\delta)$-approximate solution to $S$. 

The reduction shows that we can use a solution to the $(\epsilon,\delta/10)$-$\gcircuit$ problem with fan-out $10/\delta$ to solve the problem of $(\epsilon,\delta)$-$\gcircuit$ with unbounded fan-out. Thus the former problem is at least as hard as the latter problem. From there, we can further reduce the $(\epsilon,\delta/10)$-$\gcircuit$ problem with fan-out $10/\delta$ to the $(\epsilon'',\delta'')$-$\gcircuit$ problem for some constant $\epsilon'',\delta''>0$ with fan-out at most 2 using the same construction as in \cite{Rub18-inapprox}.
\end{proof}

\begin{remark*}
\cite{Rub18-inapprox} gives a reduction to constant fan-out for $\epsilon$-$\gcircuit$ by constructing a binary tree of logical gates to copy an output value. This reduction does not directly apply to the $(\epsilon,\delta)$-$\gcircuit$ problem with unbounded fan-out because a single corrupted gate at the top levels of the binary tree can affect $O(n)$ copied values at the leaves.
\end{remark*}

\section{Hardness of Approximating HZ Equilibria with Restriction}
\label{sec:hard with restriction}

\begin{definition}[Approximate HZ Equilibria with Restriction]\label{def:epsRHZ}
Given $\epsilon>0$, a pair $(x,p)$, where $x\in \RR^{n\times n}_{\geq 0}$ and $p\in \RR^{n}_{\geq 0}$ is an $\epsilon$-approximate HZ equilibrium with restriction ($\epsilon$-$\ahzr$) of an HZ market $M$ if:
\begin{enumerate}
    \item Unit supply: $\sum_{i\in [n]}x_{ij}=1, \forall j$
    \item Unit demand: $\sum_{j\in [n]}x_{ij}=1, \forall i$
    \item Normalized prices: $\min_{j\in [n]}p_j=0$
    \item Budget is approximately 1: $\sum_{j\in [n]}p_j x_{ij}\leq 1+\epsilon, \forall i$
    \item Bundle is approximately optimal: $\sum_{j\in [n]}u_{ij}x_{ij}\geq val_p(i)-\epsilon, \forall i$, where $val_p(i)$ is the value of the best affordable bundle for agent $i$ under unit demand.
    \item Do not buy positive-price, 0-utility goods: $\sum_{j\in [n], u_{ij}=0}p_j x_{ij}=0, \forall i$
\end{enumerate}
\end{definition}
 This definition is stricter than the definition in \cite{CCPY21} in that we additionally require condition 6---no zero utility goods are purchased with positive prices. This additional requirement allows us to derive a stronger inapproximability result shown below in \Cref{thm:ppad-constant}.

\begin{lemma}
    An $\epsilon$-$\ahzr$ exists for any $\epsilon\geq 0$.
\end{lemma}
\begin{proof}
    Since an HZ equilibrium exists when each individual purchases their cheapest best affordable bundle \citep{HZ1979}, and any such HZ equilibrium (after normalizing the minimum price to 0) satisfies the conditions 1-6 (in particular, it satisfies condition 6 since the cheapest best bundle will never buy positive-price, 0-utility goods), we conclude that an $\epsilon$-approximate HZ Equilibrium exists even when $\epsilon=0$.
\end{proof}

\begin{theorem}
\label{thm:ppad-constant}
    Finding an $\epsilon$-$\ahzr$ with restriction is PPAD-hard for some constant $\epsilon>0$.
\end{theorem}

We follow the same outline of proof as \cite{CCPY21}, but modify the lemmas therein to obtain the hardness result for $\epsilon$ as a constant, instead of $1/poly(n)$, for the restricted version. 

Similar to \cite{CCPY21}, we reduce from the problem of finding a $\kappa$-approximate equilibrium for the threshold game $H=(V,E)$. To make the proof self-contained, we repeat the setup of \cite{CCPY21} here. 

\begin{definition}[Threshold game \citep{PP21}]
\label{def:threshold}
    A threshold game is defined on a directed graph $H=(V,E)$ and a threshold $t$ ($0<t<1$). The vertices of the graph represent players with strategy space $x_v\in [0,1]$. A strategy profile $x=(x_u)_{u\in V}\in [0,1]^{|V|}$ is an $\kappa$-approximate equilibrium ($\kappa<t<1-\kappa$) if it satisfies
    \begin{align*}
        x_v = 
        \begin{cases}
            [0,\kappa] &\sum_{u\in N_v}x_u>t+\kappa\\
            [1-\kappa,1] &\sum_{u\in N_v}x_u<t-\kappa\\
            [0,1] &\sum_{u\in N_v}x_u\in [t-\kappa, t+\kappa]\\
        \end{cases}
    \end{align*}
    where $N_v$ is the set of vertices with incoming edges $(u,v)\in E$ to vertex $v$. Intuitively, each vertex needs to behave like an indicator of whether the incoming flows exceed a certain threshold, and can act arbitrarily near the threshold.
\end{definition}

\cite{PP21} showed that there exists a constant $\kappa$ such that the problem of finding $\kappa$-approximate equilibrium in a threshold game is PPAD-hard, and the hardness already holds when in-degree and out-degree of each vertex is bounded by some constant. Therefore we assume $H$ has bounded in-degree and out-degree. Additionally, the threshold is set to $1/2$ in their reduction, so we will also set $t=1/2$. 

We use the same construction of an HZ market $M_H$ as in \cite{CCPY21}, which is formally stated in Section \ref{sec: conditional hardness without restriction}. Let $(x,p)$ be an $\epsilon$-$\ahzr$ of this market $M_H$, where we set $\epsilon = 1/m^{20}$. Let $\alpha_i, \mu_i$ be the optimal solutions for the dual program \ref{def:dual-LP} for each agent $i$. Let $val_p(i)$ denote the optimal value of both of these LPs using prices $p$ from the $\epsilon$-\ahzr. 

First we make the simple observation that the dummy goods must have price 0.
\begin{lemma}
 $p(G_D)=\bar{p}(G_D)=0$   
\end{lemma}
\begin{proof}
    If any dummy good has positive price, no agent will buy it, so there is zero demand for the dummy good, contradicting that $(x,p)$ is an $\epsilon$-$\ahzr$. This also shows that any $\epsilon$-equilibrium for this market that satisfies condition 1,2,4-6 automatically satisfies condition 3, that $\min_j p_j=0$.
\end{proof}

Next we show that the prices of some goods in the market are bounded by some constants. This will be helpful for us to bound $\alpha_i$ in the dual of individual optimization LP. (We will later show that actually all goods are bounded by some constants.) Let $e=(u,v)$ unless indicated otherwise.
\begin{lemma}
\label{lem:gegv3-bdd}
\begin{align}
    p(G_e)&\leq 2+O(1/m^4)\\
    p(G_{v,3})&\leq 5/2+O(1/m^7)
\end{align}

\begin{proof}
    First we observe that any good in $G_e$ must have positive prices, i.e. $p(G_e)>0$ since otherwise each agent $A_{e,*}$ will buy at least $(1-\epsilon)$ unit of $G_e$. However, there are $64m^5$ agents in $A_{e,*}$ but only $32m^5$ goods in $G_e$, demand is greater than supply, a contradiction. 
    Thus the prices of $G_e$ is positive. This also implies that only agents with positive utility for them will be assigned these goods. Since all goods in $G_e$ need to be sold, and can be sold only to agents in $A_{e,*}, A_{e,2,\ell}, A_{e,3,\ell}, A_{e,4,\ell}$ who have positive utility for it, where each of them has a budget of $1+\epsilon$, therefore we have
    \begin{equation*}
        p(G_e)\leq \frac{(64m^5 + 50m)(1+\epsilon)}{32m^5} = 2+O(1/m^4).
    \end{equation*}
    
    Similarly, $G_{v,3}$ have positive prices and therefore can only be sold to agents in $A_v$ and $A_{e,1}$ where $e=(v,k)$ for some $k$ (there can only be a constant number of such edges since the out-degree of each vertex is bounded). Thus by market clearance, we have
    \begin{equation*}
        p(G_{v,3})\leq \frac{(5m^{10} + \Theta(48m^3))(1+\epsilon)}{2m^{10}} = 5/2 + O(1/m^7).
    \end{equation*}
\end{proof}
\end{lemma}

We also have the following lower bounds for prices of $G_e$ and $G_{v,3}$. 
\begin{lemma}[\cite{CCPY21}, Lemma 3.13]
\label{lem:gegv3-expensive}
\begin{align}
    &p(G_e)\geq 2(1-2\epsilon)\\
    &p(G_{v,3})\geq 5/3
\end{align}
\end{lemma}

\begin{corollary}
\label{cor:ge-price}
$\forall e\in E, p(G_e)=2\pm O(1/m^4)$
\end{corollary}
\begin{proof}
    By Lemma \ref{lem:gegv3-bdd}, Lemma \ref{lem:gegv3-expensive}, and $\epsilon=1/m^{20}$.
\end{proof}

\begin{lemma}
    Each agent $i$ in $M_H$ has optimal value $val_p(i)\leq 0.9$. More precisely, $val_p(A_v)<0.6$, $val_p(A_{e,*})\leq 0.5+\epsilon$, and $val_p(A_e)\leq 0.8$. 
\end{lemma}
\begin{proof} Proof of Lemma 3.13 in \cite{CCPY21} gives $val_p(A_v)<0.6$ and $val_p(A_{e,*})\leq 0.5+\epsilon$. We verify the claim for $A_e$ agents (adapted from proof of Lemma 3.14 in \cite{CCPY21} for constant $\epsilon$). Note that only $G_e$ and $G_{u,3}$ goods give them utility 1, and the other goods give them utility at most 1/2. But $G_e$ and $G_{u,3}$ are expensive as shown in Lemma \ref{lem:gegv3-expensive} and the utility they can get from them (subject to budget 1 because we are looking at the LP, not their real budget) is at most $\max\{1/2(1-2\epsilon), 3/5\} = 0.6$ when $\epsilon$ is a sufficiently small constant. Thus $val_p(A_e)\leq 0.6+0.4*0.5= 0.8$.
\end{proof}

Now we are ready to bound $\alpha_i$ and $\mu_i$ for any agent $i$. (The bound is tighter than \cite{CCPY21}, Lemma 3.3 because we utilize the fact that the prices of all goods in the market are bounded by some constant.)
\begin{lemma}
\label{lem:bound-on-mu-alpha}
\begin{align}
    &\mu_i\geq 0\\
    & 1/20 \leq \alpha_i\leq 0.9
\end{align}
\end{lemma}
\begin{proof}
    (Adapted from proof of Lemma 3.3 in \cite{CCPY21}). We have $\mu\geq 0$ because there exists a good $\ell$ with $p_\ell=0$ such as the dummy goods and the corresponding dual constraint \ref{eq:dual-constraint} implies $\mu\geq 0$. Then $\alpha_i\leq \alpha_i+\mu_i=val_p(i)\leq 0.9$. Now, since every agent $i$ has a utility 1 good, which is either $G_{e}$ or $G_{v,3}$ for some $e\in E, v\in V$, and we know their minimum prices are at most 3 by Lemma \ref{lem:gegv3-bdd}. Thus by \ref{eq:dual-constraint}, we have $\alpha_i*3+\mu_i\geq 1$ and $\alpha_i+\mu_i=val_p(i)\leq 0.9$, which yields $\alpha_i\geq 1/20$.
\end{proof}

\begin{lemma}
\label{lem:maxp-minp}
    For all groups of goods $G_j$, $\bar{p}(G_j)\leq p(G_j)+\Theta(1/m^9)$.
\end{lemma}
\begin{proof}
    By way of contradiction, suppose there exists a good in $G_j$ with price $\bar{p}(G_j)\geq p(G_j)+\Theta(1/m^9)>0$. Then this good can only be assigned to agents with positive utility for it. We note that any good has at most $\Theta(m^{10})$ agents who have positive utility for it. We also have for any such agent $i$, $\alpha_i\bar{p}(G_j)+\mu_i\geq u_{ij}$ and $\alpha_i p(G_j)+\mu_i\geq u_{ij}$. By our assumption and the fact that $\alpha_i\geq 1/20$ by Lemma \ref{lem:bound-on-mu-alpha}, we know that this good is $1/20*\Theta(1/m^9)$-suboptimal. By Lemma \ref{lem:suboptimal-alloc}, we conclude the total allocation of this good can be at most $\Theta(m^{10})*\frac{2\epsilon}{(1/20)\Theta(1/m^9)} = 40\epsilon *\Theta(m^{19}) = \Theta(1/m)<1$ since $\epsilon=1/m^{20}$, a contradiction.
\end{proof}
Note that this bound on the maximum price of each group is weaker than Corollary 3.15 in \cite{CCPY21}, since we only assume $\epsilon$ is a constant rather than polynomially small in $n$.

For simplicity of notation, let $y_1,y_2,y_3$ denote the allocation of $G_{v,1}, G_{v,2}, G_{v,3}$ to $A_v$ respectively, $u_1,u_2,u_3$ be the corresponding utilities, and $q_1,q_2,q_3$ be the respective minimum prices of each group. Recall that $u_1=\frac{1}{2m^2-1}, u_2=\frac{m^2+1}{4m^2-2}, u_3=1$.
\begin{lemma}
\label{lem:gvl-alloc}
\begin{align}
    &y_1\geq m^{10}-O(m^3)\\
    & y_2, y_3\geq 2m^{10}-O(m^3)
\end{align}
\end{lemma}
\begin{proof}
(Adapted from proof of Claim 3.17 in \cite{CCPY21}, for constant $\epsilon$). Let $\alpha,\mu$ be an optimal solution to the dual LP for agents in $A_v$.
\begin{itemize}
    \item Case 1: $\mu\geq u_1/2=\Omega(1/m^2)$. Then any good that $A_v$ has 0 utility for is $\mu$-suboptimal for them. Thus the total allocation of $A_v$ to $\overline{G_v}$ (i.e. not $G_v$ goods) is bounded by $5m^{10} * \frac{2\epsilon}{\Omega(1/m^2)} = O(m^{12})\epsilon<1$ since $\epsilon\leq 1/m^{20}$. Thus $y_1+y_2+y_3\geq 5m^{10}-1$. Since $y_1\leq |G_{v,1}|=m^{10}+S_v=m^{10}+\Theta(m^3)$ and $y_2,y_3\leq 2m^{10}$, we have $y_1\geq m^{10}-1$ and $y_2,y_3\geq 2m^{10}-O(m^3)$.
    \item Case 2: $\mu < u_1/2$. Consider the dual of LP for $A_v$ agents. For $\ell\in [3]$, since $\alpha q_\ell+\mu\geq u_\ell\geq u_1$, we have $\alpha q_\ell\geq u_1/2$. By Lemma \ref{lem:bound-on-mu-alpha} $\alpha<1$, thus $q_\ell\geq \frac{u_1}{2\alpha}\geq \frac{u_1}{2}=\Omega(1/m^2)>0$. Thus $G_{v,\ell}$ can only be allocated to positive utility agents, which are $A_v$ and $A_e$ agents where $e=(u,v)$ or $e=(v,k)$ for some $u$ and $k$. There are $\Theta(m^3)$ such $A_e$ agents since the in-degree or out-degree of $H$ is bounded by a constant. Thus the total allocation of $G_v$ to $A_v$ is $y_1+y_2+y_3\geq 5m^{10}+S_v-\Theta(m^3)\geq 5m^{10}-\Theta(m^3)$. Thus $y_1\geq m^{10}-O(m^3)$ and $y_2,y_3\geq 2m^{10}-O(m^3)$.
\end{itemize}
\end{proof}

We can also bound each agent's payment. Note that by our definition of $\ahzr$, all payment is toward positive-utility goods. For what follows, we use $p_j$ to denote the price of an individual good $j$, not the minimum price of the groups of goods that $j$ belongs to.

\begin{lemma}
\label{lem:agent-payment-bdd}
    For any agent $i$, 
    \begin{align}
        \sum_{j\in [n]}p_j x_{ij} = \sum_{j\in [n], u_{ij}>0}p_j x_{ij}= 1\pm O(\epsilon)
    \end{align}
\end{lemma}
\begin{proof}
    By \ref{eq:dual-constraint}, we have 
    \begin{align*}
        &\alpha_ip_j+\mu_i\geq u_{ij}\\
        \implies & \sum_{j\in [n]}\alpha_ip_jx_{ij} +\sum_{j\in[n]}\mu_i x_{ij}\geq \sum_{j\in [n]}u_{ij}x_{ij} \geq val_p(i)-\epsilon\\
        \implies &\alpha_i\sum_{j\in [n]}p_jx_{ij}+\mu_i\geq \alpha_i+\mu_i-\epsilon\\
       \implies & \sum_{j\in [n]}p_j x_{ij}\geq 1-\frac{\epsilon}{\alpha_i}\\
    \end{align*}
    Then using $\alpha_i\geq 1/20$ in Lemma \ref{lem:bound-on-mu-alpha}, we get 
    \begin{align*}
        \sum_{j\in [n]}p_j x_{ij}\geq 1-20\epsilon.
    \end{align*}
    Budget constraint gives the upper bound
       $\sum_{j\in [n]}p_j x_{ij}\leq 1+\epsilon$.
\end{proof}

Now we are ready to prove the following lemma.
\begin{lemma}
\label{lem:gvl-prices}
    For all $v\in V$, 
    \begin{align}
        & 0\leq p(G_{v,1})\leq \frac{1}{m^2} + O\left(\frac{1}{m^6}\right)\\
        & p(G_{v,2})=\frac{1+p(G_{v,1})}{2}\pm O\left(\frac{1}{m^7}\right)\\
        & p(G_{v,3})=2-p(G_{v,1})\pm O\left(\frac{1}{m^7}\right)
    \end{align}
    and for any $\ell\in [3]$, $G_{v,\ell}$ is at most $\delta$-suboptimal for any agent in $A_v$ for $\delta=20\epsilon$.
\end{lemma}

\begin{proof} (Adapted from proof of Lemma 3.10 in \cite{CCPY21}).
    First we note that if for any $\ell\in [3]$, $G_{v,\ell}$ is $\delta$-suboptimal for $A_v$, then we can assign at most $2\epsilon/\delta$ per agent in $A_v$ by Lemma \ref{lem:suboptimal-alloc}. Since we have to assign $y_\ell\geq m^{10}-O(m^3)$ in total by Lemma \ref{lem:gvl-alloc}, we need
    \begin{align*}
    &5m^{10}2\epsilon/\delta\geq m^{10}-O(m^3)\implies \delta\leq 10\epsilon+O(\epsilon/m^7)<20\epsilon.
    \end{align*}
    for $m$ large enough.
    Thus $G_{v,\ell}$ cannot be $20\epsilon$-suboptimal for $A_v$ for any $\ell\in [3]$. Let $\delta=20\epsilon$. This implies
    \begin{equation}
    \label{ineq: constraint-bdd}
        u_{\ell}\leq \alpha q_\ell +\mu \leq u_\ell+\delta, \forall \ell\in [3].
    \end{equation}
    Since utilities are designed so that $u_2=(3u_1+u_3)/4$, it follows that 
    \begin{align*}
        &\alpha\left(\frac{3q_1+q_3}{4}\right)+\mu-\delta\leq u_2\leq \alpha q_2+\mu\leq u_2+\delta\leq \alpha\left(\frac{3q_1+q_3}{4}\right)+\mu+\delta\\
        \implies & q_2=\left(\frac{3q_1+q_3}{4}\right) \pm O\left(\frac{\delta}{\alpha}\right)\\
        \implies  & q_2=\left(\frac{3q_1+q_3}{4}\right) \pm O(\epsilon)
    \end{align*}
    since $\alpha\geq 1/20$ by Lemma \ref{lem:bound-on-mu-alpha}.
    Now, we can use the total payment of agents in $A_v$ to bound $q_1,q_2, q_3$. We know that $q_\ell$ is the minimum price of a good in $G_{v,\ell}$, and that by Lemma \ref{lem:maxp-minp}, any good $j$ in $G_{v,\ell}$ has price $p_j=q_\ell +O(1/m^9)$. On the other hand, the payment of each agent in $A_v$ is $1\pm O(\epsilon)$ by Lemma \ref{lem:agent-payment-bdd}. Thus we have
    \begin{align*}
        &(q_1+O(1/m^9))y_1+(q_2+O(1/m^9))y_2 + (q_3+O(1/m^9))y_3 = 5m^{10}(1\pm O(\epsilon))\\
        \implies & (q_1+O(1/m^9))y_1+\left(\left(\frac{3q_1+q_3}{4}\right) \pm O(\epsilon)+O(1/m^9)\right)y_2 + (q_3+O(1/m^9))y_3 = 5m^{10}(1\pm O(\epsilon)).
    \end{align*}
    By Lemma \ref{lem:gvl-alloc} and $\epsilon=1/m^{20}$, we conclude that
    \begin{align}
        &q_2=\frac{1+q_1}{2}\pm O\left(\frac{1}{m^7}\right)\\
        & q_3=2-q_1\pm O\left(\frac{1}{m^7}\right)\label{eqn:q1q3}.
    \end{align}
    Now it remains to bound $q_1$. We first observe that $q_1<q_3$ since if not, then we have 
    \begin{align*}
        \alpha q_1+\mu\geq  \alpha q_3+\mu\geq u_3=1.
    \end{align*}
    Since $u_1 = \Theta(1/m^2)$, $G_{v,1}$ is $(1-\Theta(1/m^2))$-suboptimal for $A_v$. Thus by Lemma \ref{lem:suboptimal-alloc}, the total allocation of $G_{v,1}$ to $A_v$ is $y_1\leq 5m^{10}2\epsilon/(1-\Theta(1/m^2))<1$, a contradiction to Lemma \ref{lem:gvl-alloc}. Thus $q_1< q_3$. By \ref{ineq: constraint-bdd}, since $q_1,q_3\geq 0$, we have 
    \begin{align*}
    &q_3*u_{1}\leq q_3(\alpha q_1 +\mu) \leq q_3(u_1+\delta) \\
    &q_1*u_{3}\leq q_1(\alpha q_3 +\mu) \leq q_1(u_3+\delta)\\
        &\implies\mu\leq \frac{u_1q_3-u_3q_1+q_3\delta}{q_3-q_1} =\frac{u_1q_3-u_3q_1+O(\delta)}{q_3-q_1}.
    \end{align*}
    since $q_3\leq 3$ by \ref{eqn:q1q3}.
    Note that the denominator is positive.
    Suppose $q_1\geq 1/m^2+1/m^6$, then $q_3\leq 2-1/m^2$ by \ref{eqn:q1q3}. Then $u_1q_3-u_3q_1\leq -1/m^6$. Since $\delta=20\epsilon=O(1/m^{20})$, the numerator is negative. Thus $\mu<0$, a contradiction to $\mu\geq 0$ by Lemma \ref{lem:bound-on-mu-alpha}. Thus
    \begin{equation*}
        q_1<1/m^2+1/m^6.
    \end{equation*}
\end{proof}

Let $x^{+}(G_j, A_k)$ denote the total amount of goods in group $G_j$ that is assigned to agents in group $A_k$ such that the assignment gives positive utility to the agents. The following lemma proves a stricter bound than Lemma 3.11 in \cite{CCPY21} because of our restriction on the equilibrium. 
\begin{lemma}
\label{lem:x+vneqv}
For any $v\in V$,
    \begin{enumerate}
        \item If $x^{+}(G_v, \overline{A_v})\geq S_v+1$, then $p(G_{v,1})=1/m^2\pm O(1/m^9)$
        \item If $x^{+}(G_v, \overline{A_v})\leq S_v-1$, then $p(G_{v,1})=0$.
    \end{enumerate}
\end{lemma}
\begin{proof} (Adapted from proof of Lemma 3.11 in \cite{CCPY21}).
    By Lemma \ref{lem:gvl-prices}, we know that both $G_{v,2}$ and $G_{v,3}$ have positive prices, so no agent with 0 utility for them will be allocated these goods. 
    \begin{itemize}
        \item Case 1: $x^{+}(G_v, \overline{A_v})\geq S_v+1$.  This means the assignment of $G_{v}$ to agents with positive utility outside of $A_v$ is at least $S_v+1$. Since there are $5m^{10}+S_v$ goods in $G_v$, we must allocate at most $5m^{10}-1$ goods to $A_v$, so there exists an agent in $A_v$ who is allocated at least $1/5m^{10}$ units of goods outside of $G_v$. Let $\alpha, \mu$ be an optimal solution to the dual of LP of this agent. Since this agent has 0 utility for these goods, they are $\mu$-suboptimal, which means $1/5m^{10}\leq 2\epsilon/\mu$, so $\mu=O(m^{10}\epsilon)$. By Lemma \ref{lem:gvl-prices}, \ref{ineq: constraint-bdd} and \ref{eqn:q1q3}, we have 
        \begin{align}
            &\alpha(q_1+q_3)+2\mu =u_1+u_3+ O(\delta)\nonumber\\
            \implies & \alpha = \frac{u_1+u_3+O(\epsilon)-2*O(m^{10}\epsilon)}{q_1+q_3}\nonumber\\
             \implies & \alpha = \frac{u_1+u_3}{2}(1\pm O(1/m^7))\label{eq:case1}
        \end{align}
        since $\epsilon = 1/m^{20}$. Again by \ref{ineq: constraint-bdd}, we have 
        \begin{align*}
            q_1 &=\frac{u_1-\mu+O(\delta)}{\alpha}\\
            &=\frac{u_1}{\alpha}+\frac{-O(m^{10}\epsilon)+O(\epsilon)}{\alpha}\\
            &=\frac{u_1}{\alpha}\pm O(1/m^{10})\\
            &=\frac{u_1}{\frac{u_1+u_3}{2}(1\pm O(1/m^7))}\pm O(1/m^{10})\text{ by \ref{eq:case1}}\\
             &=1/m^2\pm O(1/m^9).
        \end{align*}
        \item Case 2: $x^{+}(G_v, \overline{A_v})\leq S_v-1$. This means the assignment of $G_{v}$ to agents with positive utility outside of $A_v$ is at most $S_v-1$. Since there are $5m^{10}+S_v$ goods in $G_v$, we must assign at least $5m^{10}+1$ goods to $A_v$ and agent with 0 utility outside of $A_v$. Since there are at most $5m^{10}$ agents in $A_v$, there is at least one agent with 0 utility for goods in $G_v$ who gets assigned some good in $G_v$. This good cannot be $G_{v,2},G_{v,3}$ since they have positive prices, so this agent must be assigned some good in $G_{v,1}$ with price 0, so $p(G_{v,1})=0$.
    \end{itemize}
\end{proof}

\begin{lemma}
\label{lem:x+ueve}
\begin{align*}
    &x^{+}(G_u, A_e)=24m^3+12m\pm O(1)\\
    &x^{+}(G_v, A_e)=24m^3+15m-6m^3p(G_{u,1})\pm O(1)\\
\end{align*}
\end{lemma}
\begin{proof}
Like Lemma \ref{lem:x+ueve-eps-delta}, the proof of Lemma \ref{lem:x+ueve} is nearly identical to the proof of Lemma 3.12 in \cite{CCPY21}.
The only modifications that need to be made are to use our bounds on the allocations and prices when they use theirs.
We again provide a mapping from their results to ours below so that this substitution is easier.
\begin{center}
\begin{tabular}{ c|c } 
    \cite{CCPY21} & Here \\ 
    \hline
    Lemma 3.5 & Lemma \ref{lem:suboptimal-alloc} \\
    Corollaries 3.6 + 3.15 & Lemmas \ref{lem:suboptimal-alloc} + \ref{lem:agent-payment-bdd} + \ref{lem:maxp-minp} \\ 
    Lemma 3.10 & Lemma \ref{lem:gvl-prices} \\
    Lemma 3.16 & Lemma \ref{lem:agent-payment-bdd}
\end{tabular}
\end{center}
As before, if we choose $\varepsilon = O(1/m^{20})$, then the fact that $O(m^6\varepsilon) = O(1/m^{14})$ for us, while $O(m^6\varepsilon) = O(1/\mathrm{poly}(n))$ for them, does not quantitatively affect the proof.

\end{proof}
Then the proof of Theorem \ref{thm:ppad-constant}
follows from Lemma \ref{lem:gvl-prices}, Lemma \ref{lem:x+vneqv}, and Lemma \ref{lem:x+ueve} as in \cite{CCPY21}. We add the proof here for completion.
\begin{proof}(Adapted from proof of Theorem 3.1 in \cite{CCPY21}. The only change is we apply a tighter bound from Lemma \ref{lem:x+vneqv} to get $p(G_{v,1})=0$ for Case 1 instead.)

For any instance of the threshold game $H=(V,E)$ where the in-degree and out-degree are bounded by some constant and the threshold is $1/2$, we can set $m=\lceil C/\kappa\rceil$ and create the HZ market $M_H$ as described before. Let $(x,p)$ be an $\epsilon$-$\ahzr$ of $M_H$ where $\epsilon=\frac{1}{m^{20}} = \frac{1}{{\lceil C/\kappa\rceil}^{20}}$. Let $x_v=\min (1,m^2p(G_{v,1}))\in [0,1]$. Then we claim that $x=(x_v)_{v\in V}$ is a $\kappa$-approximate equilibrium of $H$.
\begin{itemize}
\item Case 1: $\sum_{u\in N_v}x_u>0.5+\kappa$. By definition of $x_u$, we have $x_u\leq m^2 p(G_{u,1})$. This implies $ p(G_{u,1})\geq x_u/m^2\implies \sum_{u\in N_v}p(G_{u,1})\geq \sum_{u\in N_v}x_u/m^2>(1/m^2)(0.5+\kappa)$. Then by Lemma \ref{lem:x+ueve}, we have
\begin{align*}
    x^{+}(G_v, \overline{A_v})&\leq \outdeg(v)\cdot (24m^3+12m+ O(1)) +\sum_{u\in N_v} (24m^3+15m-6m^3p(G_{u,1})+ O(1))\\
    &= S_v+3m-6m^3\sum_{u\in N_v}p(G_{u,1})+O(1)\\
    &\leq S_v+3m-6m^3(1/m^2)(0.5+\kappa)+O(1)\\
    &= S_v-6m\kappa + O(1)
    <S_v-1.
\end{align*}
since $m\kappa \geq C$ for some large constant $C$ and we can choose the $C$ to be larger than the $O(1)$ term.
Then by Lemma \ref{lem:x+vneqv}, we have $p(G_{v,1})=0$, which implies $x_{v}=0$, as desired.
\item Case 2: $\sum_{u\in N_v}x_u<0.5-\kappa$. By Lemma \ref{lem:gvl-prices}, we have $p(G_{u,1})\leq 1/m^2+O(1/m^6)$.
 Thus 
\begin{align*}
&m^2p(G_{u,1})\leq 1+O(1/m^4)\\
    &\implies x_u = \min(1,m^2p(G_{u,1})) \geq m^2p(G_{u,1})-O(1/m^4)\\
    &\implies p(G_{u,1})\leq x_u/m^2+O(1/m^6)\\
    &\implies \sum_{u\in N_v}p(G_{u,1})\leq \sum_{u\in N_v}x_u/m^2+O(1/m^6)<(1/m^2)(0.5-\kappa+O(1/m^4))
\end{align*}
by bounded in-degree.
Then by Lemma \ref{lem:x+ueve}, we have
\begin{align*}
    x^{+}(G_v, \overline{A_v})&\geq \outdeg(v)\cdot (24m^3+12m- O(1)) +\sum_{u\in N_v}(24m^3+15m-6m^3p(G_{u,1})-O(1))\\
    &= S_v+3m-6m^3\sum_{u\in N_v}p(G_{u,1})-O(1)\\
    &\geq S_v+3m-6m^3(1/m^2)(0.5-\kappa+O(1/m^4))-O(1)\\
    &= S_v+6m\kappa -O(1/m^3)- O(1)
    \geq S_v+C> S_v+1.
\end{align*}
Thus by Lemma \ref{lem:x+vneqv}, $p(G_{v,1})\geq 1/m^2-O(1/m^9)$, which means $x_v\geq 1-O(1/m^7)\geq 1-\kappa$, as desired.
\end{itemize}
\end{proof}

\end{document}